\documentclass[letterpaper, 11pt]{article}

\usepackage{amsmath, amssymb, amsfonts, amsthm}
\usepackage{bm}
\usepackage{mathtools}
\usepackage{mathrsfs}
\usepackage{geometry}
\usepackage{setspace}
\usepackage[utf8x]{inputenc}
\usepackage{algpseudocode,algorithmicx,algorithm}
\usepackage[T1]{fontenc}
\usepackage{lmodern}
\usepackage{graphicx}
\usepackage{enumerate}
\usepackage[natural, dvipsnames, pdftex]{xcolor}
\usepackage{tikz}
\usepackage{verbatim}
\usepackage{hyperref}
\usepackage{cancel}
\usepackage[all]{xy}
\usepackage{mdframed}
\usepackage[capitalize]{cleveref}
\usepackage{todonotes}
\usepackage{algpseudocode}
\usepackage{comment}
\usepackage{booktabs}
\usepackage{tabularx}
\usepackage{array}

\usepackage[numbers]{natbib}

\usepackage{enumitem}

\usepackage{indentfirst}
\theoremstyle{plain}
\newtheorem{thm}{Theorem}[section]
\newtheorem{lemma}[thm]{Lemma}
\newtheorem{prop}[thm]{Proposition}
\newtheorem{cor}[thm]{Corollary}

\theoremstyle{definition}
\newtheorem{defn}[thm]{Definition}

\theoremstyle{remark}
\newtheorem{rmk}[thm]{Remark}

\newcommand{\F}{\mathbb{F}}

\newcommand{\cC}{\mathcal{C}}

\def\g{{\mathfrak{g}}}

\begin{document}

\title{On Proximity Gaps for One-Point Algebraic Geometric Codes}

\author{Songsong Li\thanks{Shanghai Jiao Tong University, Shanghai, China. \texttt{songsli@sjtu.edu.cn}.}
\and  Chaoping Xing \thanks{Shanghai Jiao Tong University, Shanghai, China. \texttt{xingcp@sjtu.edu.cn}.}
\and Chen Yuan\thanks{Shanghai Jiao Tong University, Shanghai, China. \texttt{chen\_yuan@sjtu.edu.cn}.}
\and Ruiqi Zhu\thanks{Shanghai Jiao Tong University, Shanghai, China. \texttt{sjtuzrq7777@sjtu.edu.cn}.}
}
\date{}
\maketitle

\begin{abstract}
    Proximity gaps and correlated agreement underpin the soundness of interactive oracle proofs. We establish these properties for general one-point algebraic geometry (AG) codes, extending corresponding results for Reed-Solomon codes. For codes of length $n$ over $\mathbb F_q$ with designed distance $n-m$, write $\rho=m/n$ and fix $\rho\in(0,1)$. For polynomial curves of degree at most $\ell$ and any $0<\varepsilon<1$, we prove correlated agreement with multiplicative slack $\varepsilon$ and soundness $O(\ell(\ell+1)/(\varepsilon q))$ whenever $e\le(1-\varepsilon)(n-m)/2$, with no additional genus penalty. This gives asymptotically good AG code families from Garcia--Stichtenoth tower of unbounded length over a fixed field with soundness bounds independent of $n$, beyond the RS restriction $n\le q$. We also establish mutual correlated agreement up to radius $(1-\sqrt{\rho}-\eta)n$, with soundness $O(\ell n/(\eta^3q))$ for $0<\eta<1-\sqrt{\rho}$. Our half designed distance proof combines auxiliary divisors, simultaneous Berlekamp-Welch interpolation, and error-incidence counting to obtain a common polynomial representation without a Polishchuk-Spielman divisibility lemma. 
    
    Finally, two counterexamples on the Garcia-Stichtenoth tower show limitations beyond half the minimum distance, at radii below the Johnson bound and with positive slack. The first requires a positive constant soundness lower bound independent of $q$ and $n$, precluding a uniform $O(1/q)$ guarantee throughout the Johnson regime. The second fixes the field and lets the relative gap above half distance tend to zero while ruling out correlated agreement for every soundness parameter $\sigma<1$, with positive limiting rate and relative error radius.
\end{abstract}

\section{Introduction}\label{sec:introduction}
%the background of proximity gaps and short mention of its applications in proof systems and why this an active area. 
Random linear combinations are a basic soundness tool in algebraic proof systems. Let $\mathcal{C}\subseteq\F_{q}^n$ be a linear code. For any two words $\mathbf{u}_0, \mathbf{u}_1\in\F_q^n$, let $\mathbf r(z)=\mathbf u_0+z\mathbf u_1$ be a random linear combination for $z\in\F_q$. If $\mathbf{u}_0, \mathbf{u}_1$ agree with codewords $\mathbf{c}_0, \mathbf{c}_1$ on the same set of coordinates, then every $\mathbf r(z)$ agrees with $\mathbf{c}_0+z\mathbf{c}_1$ on that set. The converse is less straightforward since errors may cancel for some values of $z$ and a codeword close to $\mathbf r(z)$ may depend on $z$. A proximity gap at a given radius says that either every point on the line is close to the code, or only a small fraction of the parameters give points close to the code. %Correlated at a given radius says that either every point on the line is close to the code, or only a small fraction of the points are close.
Correlated agreement strengthens this conclusion by requiring codewords approximating the input words on a common set of coordinates. Mutual correlated agreement further controls the individual agreement sets of nearby codewords. These stronger properties are useful when a protocol must preserve not only proximity, but also the coordinates on which agreement holds~\cite{BCIKS23,WHIR25}.

This question arises in interactive oracle proofs of proximity (IOPPs). Random linear combinations allow a verifier to batch proximity claims before making a small number of oracle queries~\cite{RVW13,AHIV17,BKS18}. If the two input words have no nearby codewords on a common coordinate set, but $\mathbf{r}(z)$ is close to $\cC$ for at most $B$ choices of $z$, this randomization step contributes at most $B/|\F_q|$ to the soundness error. Proximity gaps therefore provide soundness guarantees and guide parameter choices for IOPP protocols. FRI and DEEP-FRI made it important to understand whether random linear combinations preserve distance from Reed-Solomon codes~\cite{FRI18,DEEP20}. More recently, STIR and WHIR have used proximity and correlated agreement guarantees to reduce query complexity and improve verification time~\cite{STIR24,WHIR25}. Through polynomial commitments, these improvements also support more efficient succinct proof systems.

For arbitrary linear codes with relative minimum distance $\Delta$, proximity gaps for affine lines hold below relative radius $\Delta/3$~\cite[Section~2.1]{DP24}. The \emph{1.5 Johnson radius} $1-(1-\Delta)^{1/3}$ was established with proximity loss in~\cite{DEEP20}. Gao et al.~\cite{GKL24} obtained soundness $O(n/q)$ below this radius without additional proximity loss. %for fixed relative parameters.

For Reed-Solomon codes, Ben-Sasson et al.~\cite{BCIKS23} established proximity gaps below half the designed distance and up to the Johnson radius. Ben-Sasson et al.~\cite{BCHKS25} later improved the soundness error from $O(n^2/q)$ to $O(n/q)$ below the Johnson bound, without additional proximity loss. They also obtained bounds independent of the block length below half the designed distance when slack is allowed.

Beyond half the minimum distance, however, the RS counterexample of~\cite{DEEP20}, recalled in~\cite[Section~1.1]{BCHKS25}, shows that correlated agreement can require soundness $\Omega(n/q)$ even below Johnson and with a small fixed positive slack. In this example, $\mathbf u_0+z\mathbf u_1$ is close to the code for $\Omega(n)$ values of $z$, but the words $(\mathbf{u}_0,\mathbf{u}_1)$ do not have the common agreement required by correlated agreement.

Recent work has extended list decoding of Reed-Solomon codes all the way to capacity. Brakensiek et al.~\cite[Corollary~5.1]{BCPZZ26} gave deterministic list decoding algorithms up to radius $1-R-\varepsilon$ over sufficiently large prime fields and arbitrary evaluation sets, with running time polynomial in $q$ for fixed $R,\varepsilon$. For fixed $R$ and $q=n^{O(1)}$, however, the resulting upper bound on the list size is $n^{(1/\varepsilon)^{O_R(1/\varepsilon)}}$. Jeronimo~\cite{Jer26} further obtained capacity list decoding over prime fields together with mutual correlated agreement at the same radius, without additional proximity loss, with list size $n^{O_\varepsilon(1)}$ and soundness error $n^{O_\varepsilon(1)}/q$. Very recently, Dao et al.~\cite{DKT26} further improved the list size and mutual correlated agreement soundness beyond the Johnson radius and gave better dependence on $\varepsilon$ up to capacity.

A complementary question concerns the relation between alphabet size and block length. Below half the designed distance, RS codes have soundness $O(1/q)$ when the relative parameters and the slack are fixed. This bound does not require $q$ to grow with $n$. The construction of RS codes, however, needs $n$ distinct evaluation points in $\mathbb F_q$, so $n\leq q$. Thus, realizing such a guarantee at unbounded lengths over a fixed field requires a different code family.

Algebraic-geometric (AG) codes provide a natural candidate. They generalize RS codes by evaluating functions with bounded poles at rational points of an algebraic curve. Over a sufficiently large fixed square field, towers of curves yield one-point AG codes of unbounded length with positive rate and positive relative distance~\cite{Sti09,GarciaStichtenoth95}. The RS proximity proofs, however, use properties of polynomials that do not apply directly to functions on a curve.

To prove correlated agreement near half the designed distance, the RS argument shows that many nearby messages have the form $f_z=v_0+zv_1$, where $v_0,v_1$ do not depend on $z$. The proof uses the Polishchuk-Spielman lemma, which is a divisibility result for bivariate polynomials with degree bounds. For AG codes, messages are functions on a curve with bounded poles, so the polynomial divisibility lemma does not apply directly. We therefore need a different argument to prove the relation $f_z=v_0+zv_1$.

In the Johnson proof for RS codes, the Taylor expansion of a message polynomial is finite. Once its coefficients are known, the polynomial can be evaluated at every code coordinate. For AG codes, a function's expansion at one point is generally infinite, and a finite truncation cannot simply replace the function when evaluating at other points. We therefore need to recover these function values from finitely many expansion coefficients and use them to obtain the same order of soundness error. This leads to the following question: To what extent do the RS proximity and agreement guarantees below half the designed distance and in the Johnson regime extend to one-point AG codes?

\subsection{Our contributions and comparisons with Reed-Solomon codes}
\label{subsec:intro-results}

\paragraph{Parameters and agreement notions.}
Let $\mathcal X/\mathbb F_q$ be a smooth projective geometrically irreducible curve of genus $\g$, and let $P_1,\ldots,P_n,P_\infty$ be pairwise distinct rational points. Fix a positive integer $m$ satisfying $2\g-1\leq m<n$, set two divisors $G=mP_\infty$ and $D=P_1+\cdots+P_n.$ The Riemann-Roch space $\mathcal L(G)$ consists of rational functions on $\mathcal X$ with no poles outside $P_{\infty}$ and pole order at most $m$ at $P_{\infty}$. We consider the one-point AG code
\[
  \mathcal C=\mathcal C(D,G)
  =\{f(D):=(f(P_1),\ldots,f(P_n)):f\in\mathcal L(G)\}.
\]
We simply call $f(D)$ the evaluations of $f$ on $D$. The assumptions on $m$ ensure that the $n$-point evaluation is injective and that code dimension is $m+1-\g$~\cite{Sti09}. Moreover, $\mathcal C$ has minimum distance at least its designed distance $d:=n-m$. Thus, $\cC$ is an $[n,m+1-\g,\geq n-m]$ linear code over $\F_q$. We write $\rho=m/n$, so that the code rate $R=\rho-\g/n+1/n$. We state our bounds in terms of $\rho$, consistent with the Johnson radius $1-\sqrt{\rho}$ for AG list decoding based on the designed minimum distance~\cite{GS99}.

Fix an integer $\ell\ge1$. We prove our results for polynomial curves of degree at most $\ell$
\[
    \mathbf r(Z)=\sum_{j=0}^{\ell}\mathbf u_jZ^j,
    \qquad \mathbf u_0,\ldots,\mathbf u_\ell\in\mathbb F_q^n.
\]
In particular, the case $\ell=1$ gives affine lines.

Following the parameter notation of~\cite{goyal2026optimal}, we use $\ell$ for the degree bound, $e$ for the Hamming radius, $\sigma$ for the soundness error and $\varepsilon$ for the multiplicative slack. Then $(\ell,e,\sigma,\varepsilon)$-correlated agreement (CA) says that, if more than a $\sigma$ fraction of challenges $z\in\F_q$ satisfy $d(\mathbf r(z),\mathcal C)\le e$, then there exist codewords $\mathbf c_0,\ldots,\mathbf c_\ell$ and a common set $I\subseteq [n]$ of size at least $n-(1+\varepsilon)e$ such that $\mathbf c_i|_I=\mathbf u_i|_I$ for every $i=0,\ldots,\ell.$ 

This implies $(\ell,e,\sigma,\varepsilon)$-proximity gaps (PG): under the same hypothesis, every point of the curve is within distance $(1+\varepsilon)e$ of the code. The stronger $(\ell,e,\sigma)$-mutual correlated agreement (MCA) requires that, outside a single exceptional set of at most $\sigma q$ challenges, \emph{every specified agreement set} of size at least $n-e$ for $\mathbf r(z)$ admits codewords matching all original words on that same set. We write $\varepsilon=0$ for CA and PG without slack and MCA is defined without slack. Formal definitions appear in Section~\ref{subsec:agreement-definitions}.

%We distinguish two parameters throughout. The slack $\varepsilon$ measures loss in the \emph{conclusion}; the Johnson margin $\eta$ measures distance from the \emph{radius boundary}. A theorem without slack has $\varepsilon=0$, even when $\eta>0$.

\begin{thm}[CA up to half the designed distance; see Theorem~\ref{thm:ag-corrlated-agreement}]
\label{thm:intro-unique}
Let $\cC$ be an $[n,m+1-\mathfrak{g}, \geq n-m]$ one-point AG code. Let $e\ge1$ be an integer and assume $e<\frac{n-m}{2}$. For every $0<\varepsilon<1$, $\mathcal C$ has $(\ell,e,\sigma_\varepsilon,\varepsilon)$-CA with $\sigma_\varepsilon=O\!\left(\max\left\{\frac{\ell}{\varepsilon q},\frac{\ell(\ell+1)n}{(n-m-2e)q}\right\}\right).$ It also has $(\ell,e,\sigma_0,0)$-CA with $\sigma_0=\min\left\{1,\max\left\{\frac{\ell(e+1)}{q},\frac{\ell(\ell+1)n}{(n-m-2e)q}\right\}\right\}.$
\end{thm}

%need more illustrations
\begin{rmk}
\begin{enumerate}
    \item[(1)] Theorem~\ref{thm:intro-unique} achieves a radius approaching half the designed minimum distance, with no additional genus dependent loss. Specifically, for fixed $\rho<1$ and $0<\varepsilon<1$, whenever $e\le (1-\varepsilon)\frac{n-m}{2},$ the soundness error is bounded by $O\!\left(\frac{\ell(\ell+1)}{\varepsilon q}\right).$  For fixed curve degree $\ell$ and slack $\varepsilon$, this bound is independent of the block length $n$. Consequently, the soundness requirement alone does not force the field size to grow with $n$, allowing a prescribed soundness error over a fixed field whenever an asymptotically good family of AG codes is available.
\item[(2)] Without slack, for affine lines, if both {$e/n$} and $\frac{n-m-2e}{n}$ are fixed positive constants, the soundness bound becomes $\min\{1,(e+1)/q\}$ for all sufficiently large $n$. In particular, this gives a nontrivial guarantee in the regime $e+1<q<n$, extending the admissible block lengths beyond the restriction $n\le q$ of RS codes. The exact finite length bounds appear in Corollaries~\ref{cor:constant} and~\ref{cor:without-slack}.
%Since $e<n/2$, \LSSnote{this bound can still be nontrivial when $q<n$, beyond the RS restriction $n\le q$. The exact finite length bounds appear in Corollaries~\ref{cor:constant} and~\ref{cor:without-slack}.}

%Without slack, if both {$e/n$} and $\frac{n-m-2e}{n}$ are fixed positive constants, the soundness bound becomes $\min\{1,(e+1)/q\}$ for all sufficiently large $n$. Since $e<n/2$, {this bound can still be nontrivial when $q<n$, beyond the RS restriction $n\le q$. The exact finite length bounds appear in Corollaries~\ref{cor:constant} and~\ref{cor:without-slack}.}
\end{enumerate}
\end{rmk}

%The radius in Theorem~\ref{thm:intro-unique} has no additional genus loss relative to half the designed distance. For fixed $\rho<1$ and $0<\varepsilon<1$, if $e \le (1-\varepsilon)\frac{n-m}{2},$ then the theorem gives a soundness error of $O(\ell(\ell+1)/(\varepsilon q))$, independent of the block length $n$ for fixed $\ell$. Thus the field size needed for a prescribed soundness error need not grow with $n$. Without slack, for fixed $\ell$, if $e/n$ and $(n-m-2e)/n$ are fixed positive constants, the soundness bound becomes $\min\{1,\ell(e+1)/q\}$ for all sufficiently large $n$. For $\ell=1$, since $e<n/2$, this bound can still be nontrivial when $q<n$, beyond the standard RS restriction $n\le q$. The exact finite-length bounds appear in Corollaries~\ref{cor:constant} and~\ref{cor:without-slack}.}

%The radius in Theorem~\ref{thm:intro-unique} has no additional genus loss relative to half the designed distance. For fixed positive slack and a fixed positive margin $n-m-2e$, the theorem gives a soundness error of $O(1/q)$, independent of the block length $n$. Thus the field size needed for a prescribed soundness error need not grow with $n$. Without slack, if $e/n$ and $\frac{n-m-2e}{n}$ are fixed positive constants, the soundness bound becomes $\min\{1,(e+1)/q\}$ for all sufficiently large $n$. Since $e<n/2$, this bound can still be nontrivial when $q<n$, beyond the standard RS restriction $n\le q$. The exact finite-length bounds appear in Corollaries~\ref{cor:constant} and~\ref{cor:without-slack}.

% short introduction of GS tower.
We next instantiate Theorem~\ref{thm:intro-unique} with
asymptotically good AG codes over a fixed alphabet. The Garcia-Stichtenoth (GS) tower~\cite{GarciaStichtenoth95} provides an explicit sequence of curves $\mathcal X_i/\F_q$, for square $q$, whose genera $\mathfrak{g}_i$ tend to infinity and whose numbers of rational points $N_i$ satisfy $\lim_{i\to\infty}\frac{N_i}{\g_i}=\sqrt q-1$. Choose one rational point as $P_{\infty,i}$ and use the remaining $n_i=N_i-1$ rational points for evaluation. Then $ \lim_{i\to\infty}\frac{\mathfrak{g}_i}{n_i}=\frac{1}{\sqrt q-1}$. 

Fix a square prime power $q>9$ and choose $\frac{2}{\sqrt q-1}<\rho<1.$ Setting $m_i=\lfloor\rho n_i\rfloor$ ensures that
$2\g_i-1\le m_i<n_i$ for all sufficiently large $i$.
The resulting AG codes $\mathcal C_i$ satisfy
\[
    \lim_{i\to\infty}R_i
    =\rho-\frac{1}{\sqrt q-1}>0,\quad 
    \liminf_{i\to\infty}
    \frac{d(\mathcal C_i)}{n_i}
    \ge 1-\rho>0.
\]
%Thus, these codes form an asymptotically good family over the fixed field $\F_q$. 
Applying Theorem~\ref{thm:intro-unique} to this family yields the following corollary.% with a positive relative proximity radius and, for fixed $\ell$ and $\varepsilon$, a soundness bound uniform in the block length.

\begin{cor}
\label{cor:intro-fixed-alphabet}
Fix an integer $\ell\ge1$, $\rho\in(0,1)$, $0<\gamma<(1-\rho)/2$ and $0<\varepsilon<1$. Let $q$ be a fixed square prime power satisfying $\rho>\frac{2}{\sqrt q-1}$. There are one-point AG code families over $\mathbb F_q$ of unbounded length, with $m/n\longrightarrow\rho$ and $e=\lfloor\gamma n\rfloor$, that have $(\ell,e,\sigma_\varepsilon,\varepsilon)$-CA for soundness $\sigma_\varepsilon$ independent of $n$. Hence, for any prescribed constant soundness target $\sigma>0$, $q$ can be chosen depending on $\ell,\rho,\gamma,\varepsilon,\sigma$ but not on $n$.
\end{cor}

The following theorem extends the Johnson range MCA for RS codes to general one-point AG codes. Compared with Theorem~\ref{thm:intro-unique}, it reaches a larger radius at the cost of a soundness error bound that depends on the block length.

\begin{thm}[MCA up to Johnson bound; see Theorem~\ref{thm:ag-johnson}]
\label{thm:intro-johnson}
For fixed $\rho\in (0,1)$, let $0<\eta<1-\sqrt\rho$. For every integer $0\le e\le (1-\sqrt\rho-\eta)n,$ $\mathcal C$ has $(\ell,e,\sigma)$-MCA with $\sigma=O\!\left(\frac{\ell n}{\eta^3q}\right).$
\end{thm}

%Moreover, the soundness order matches the current RS bound of Dao et al.~\cite[Theorem~5.12]{DKT26}.

\paragraph{Limitations beyond half the minimum distance.}
The block length independent $O(1/q)$ soundness error for CA below half the designed distance cannot, in general, be extended to the Johnson regime, even when positive multiplicative slack is allowed. 

%Section~\ref{sec:counterexample} demonstrates this limitation for affine lines on codes of positive rate from the Garcia--Stichtenoth tower. 
\begin{thm}[Corollary~\ref{cor:counterexample}]
\label{thm:intro-trace-asymptotic}
Fix a prime $p$, a slack $0\le\varepsilon<1$, and $\zeta>0$.
For every sufficiently large integer $a$, there exists $h_0$
such that, for every integer $h\ge h_0$, the
Garcia-Stichtenoth tower over $\mathbb F_q$, with
$q=p^{2ah}$, yields a family of one-point AG codes
$\mathcal C_j$ of positive asymptotic rate, with block length $n_j\to\infty$. For all sufficiently large $j$, there exists an integer $e_j\ge0$ such that 
\[
0<
\frac{e_j}{n_j}-\frac{d(\mathcal C_j)}{2n_j}
<\zeta,\quad \frac{e_j}{n_j}<1-\sqrt{\rho_j}.
\]
Nevertheless, if $\mathcal C_j$ has
$(1,e_j,\sigma,\varepsilon)$-CA, then $\sigma\ge\frac{q-1}{(p^a-1)q}\ge\frac{1}{2(p^a-1)}.$
\end{thm}
Fix an admissible $a$ and let $h\to\infty$. Then
$q=p^{2ah}\to\infty$, while the required soundness $\sigma\ge 1/(2(p^a-1))>0$. Thus these counterexamples rule out a uniform $O(1/q)$ soundness guarantee throughout the Johnson regime.

\begin{thm}\label{thm:intro-trace}
Let $p$ be a prime, let $a\ge1$ and $h\ge2$ be integers,
and set $s=p^a$ and $q=s^{2h}$.
At any level of the Garcia-Stichtenoth tower, with parameters
$n,b$ as defined in Section~\ref{sec:GS-tower}, set $m=\left(1-\frac{2}{s}+\frac{1}{s^3}\right)qb$, $e=n-\left(1-\frac{1}{s}\right)qb,$ and let $\mathcal C=\mathcal C_{\mathcal L}(D,mP_\infty)$ and $\rho=m/n$. The radius $e$ satisfies 
\[
0<
\frac en-\frac{d(\mathcal C)}{2n}
\le\frac{1}{2s^3},\quad \frac en<1-\sqrt{\rho}.
\]
Nevertheless, there exists an affine line
$\mathbf r(z)=\mathbf u_0+z\mathbf u_1$ such that $d(\mathbf r(z),\mathcal C)=e$ for at least $\frac{q-1}{s-1}$ distinct $z\in\mathbb F_q$, while $d([\mathbf u_0,\mathbf u_1],\mathcal C^2)\ge\left(2-\frac1s\right)e.$ Consequently, for every $0\le\varepsilon<1-1/s$,
any $(1,e,\sigma,\varepsilon)$-CA guarantee requires $\sigma
\ge\frac{1}{2(s-1)}.$
\end{thm}
Taking $a$ large makes the relative gap above
$\frac{d(\mathcal C)}{2n}$ arbitrarily small and
accommodates any prescribed slack $\varepsilon<1$.
Once $p,a$ are fixed, increasing $h$ makes $q$ arbitrarily
large, while the soundness lower bound remains a positive
constant. Thus the $O(1/q)$ guarantee below half the designed
distance cannot extend throughout the Johnson
regime, even with positive slack. This does not contradict
our Johnson-radius bound, which depends on $n$.

%The first construction shows that $(1,e,\sigma,\varepsilon)$-CA requires $\sigma\ge c>0$, where $c$ is independent of $q$ and $n$ but may depend on the chosen gap and slack. The relative radius can exceed half the relative minimum distance by an arbitrarily small positive amount while remaining below the Johnson radius (). A second construction keeps the field fixed while this gap tends to zero and requires $\sigma=1$ for $(1,e,\sigma,\varepsilon)$-CA (Corollary~\ref{cor:GS-fixed-field}). Thus the $O(1/q)$ guarantee does not extend throughout the Johnson regime for these codes.

\paragraph{Comparisons.} Table~\ref{tab:rs-ag-comparison} compares our results for one-point AG codes with the corresponding bounds for RS codes, using the common parameter $\rho=m/n$. Our results preserve the displayed error radii and the soundness bounds for zero slack CA and MCA. For CA with multiplicative slack, our upper bound incurs an additional factor of $\ell+1$, which is constant for fixed curve degree. Note that RS codes require $n\le q$. In contrast, for fixed $\ell$ and positive slack, suitable AG code families allow unbounded $n$ over a fixed field while maintaining a fixed soundness bound. 
\begin{table}[t]
\centering
\setlength{\tabcolsep}{6pt}
\renewcommand{\arraystretch}{1.15}
\begin{tabularx}{\textwidth}
  {@{}l*{2}{>{\centering\arraybackslash}X}@{}}
\toprule
& \textbf{RS} & \textbf{AG (this work)} \\
\midrule
$(\ell,e,\sigma,\varepsilon)$-CA
& \cite{BCHKS25}
& Theorem~\ref{thm:intro-unique} \\

Error radius
& $e<(n-m)/2$
& $e<(n-m)/2$ \\

Soundness with slack
& $O(\ell/(\varepsilon q))$
& $O(\ell(\ell+1)/(\varepsilon q))$ \\

Soundness without slack
& $\ell(e+1)/q$
& $\ell(e+1)/q$ \\

\midrule
$(\ell,e,\sigma)$-MCA
& \cite{DKT26}
& Theorem~\ref{thm:intro-johnson} \\

Error radius
& $e\le (1-\sqrt{\rho}-\eta)n$
& $e\le (1-\sqrt{\rho}-\eta)n$ \\

Soundness
& $O(\ell n/(\eta^3q))$
& $O(\ell n/(\eta^3q))$ \\
\bottomrule
\end{tabularx}
\caption{
The comparison is for polynomial curves of degree at most $\ell$, with fixed $\rho\in(0,1)$. Affine lines correspond to $\ell=1$. The implicit constants may depend on $\rho$, but not on $\ell$. In the CA row with slack, we additionally assume $e\le(1-\varepsilon)(n-m)/2$ and the RS bound follows from the proof in~\cite[Section~4.1]{BCHKS25}. In the CA row without slack, the displayed asymptotic bound applies when $\ell$ is fixed, $e/n$ and $(n-m-2e)/n$ are fixed positive constants, and $n$ is sufficiently large. In the MCA row, the RS bound follows from~\cite[Sections~5.7 and~7.2]{DKT26}.
%All soundness bounds are understood to be capped at $1$.
}
\label{tab:rs-ag-comparison}
\end{table}

%\begin{proof}
%Use an optimal tower, take all but one of its rational points as evaluation points, and take $m=\lfloor\rho n\rfloor$. The first inequality in~\eqref{eq:intro-fixed-q} ensures $m\ge2\g-1$ for all sufficiently large $n$. The second ensures a positive limiting value of $\kappa$. Apply Theorem~\ref{thm:intro-unique}. The displayed bound follows from $\g/n=\beta+o(1)$ and rounding. As square prime powers $q$ grow, $\beta$ tends to zero and the expression in parentheses remains bounded in terms of the fixed parameters, so any fixed $\sigma_0$ can be attained.
%\end{proof}

%Constant-alphabet proximity gaps are known for other constructions~\cite{goyal2026optimal}. Corollary~\ref{cor:intro-fixed-alphabet} obtains this tradeoff within ordinary one-point evaluation codes, without folding, concatenation, or extension-field challenges. A fixed soundness target and fixed slack are compatible with arbitrarily long codes in this AG family; for fixed $q$, the guarantee is a fixed soundness bound, rather than a bound tending to zero with $n$.

\subsection{Technique overview}

\subsubsection{Correlated agreement up to half the designed distance.}\label{subsec:intro-unique-technique}

Fix a polynomial curve $\mathbf{r}(z)=\sum_{j=0}^{\ell}\mathbf{u}_jz^j$ and an integer $0\le e<d/2$. Let $S$ be the set of challenges $z\in\F_q$ for which $\mathbf{r}(z)$ is within distance $e$ of the code. For each $z\in S$, let $f_z\in \mathcal{L}(G)$ be the unique function whose evaluation lies within distance $e$ of $\mathbf{r}(z)$. For sufficiently large $|S|$, we prove that $f_z=\sum_{j=0}^{\ell}v_jz^j$ for every $z\in S$, with $v_0,\ldots,v_\ell\in\mathcal L(G)$ independent of $z$. Then a coordinatewise counting argument yields correlated agreement. 

We first recall the proof strategy for RS codes and identify the main obstacles to extending it to general one-point AG codes.

%Our goal is to show that, when $|S|$ is sufficiently large, these functions have the common form $f_z=\sum_{j=0}^{\ell}v_jz^j$ for every $z\in S$, where $v_0,\ldots,v_\ell\in \mathcal{L}(G).$ Once this is established, a counting argument gives correlated agreement. The main step lie in the establishment of $\sum_{j=0}^{\ell}v_jz^j$. In the following we first recall recall the argument of CA for RS codes within the unique-decoding radius and highlight the main obstacles to extending CA from RS codes to general one-point AG codes.

\paragraph{Recap of the CA for RS codes up to unique decoding radius.}
We recall the argument in~\cite{BCHKS25}. In the RS case, $\mathcal X/\F_q$ is a projective line with genus $0$ and $\mathcal{L}(G)=\mathcal{L}(mP_{\infty})=\F_q[X]_{\leq m}$. The evaluation domain $D$ is supported by $n$ points in $\F_q$. The corresponding RS code is an $[n,m+1,d=n-m]$ linear code. 

The proof begins with Berlekamp--Welch (BW) equation for the formal received word $\mathbf{r}(Z)=\sum_{j=0}^{\ell}\mathbf u_jZ^j$. Over the evaluation domain $D$, for every $0\leq j\leq \ell$, there exists a unique polynomial $u_j(X)$ of degree less than $n$ such that $u_j(D)=\mathbf u_j$. Set $U(X,Z)=\sum_{j=0}^{\ell}u_j(X)Z^j$. Assume $e<d/2$. Choose $h= d-2e-1$ and $D_Z=\lfloor \frac{\ell e}{h+1}\rfloor+1$. Then there exist $A=\sum_{i=0}^{D_Z-1}a_i(X)Z^i$ and $B=\sum_{i=0}^{D_Z+\ell-1}b_i(X)Z^i$ with $\deg_X(a_i)\leq e+h$ and $\deg_X(b_i)\leq m+e+h$ and $A\neq0$ that satisfy the BW equations over $\F_q[Z]$
\[
    A(x,Z)U(x,Z)=B(x,Z)
\]
at every evaluation point $x$. For sufficiently many $z\in S$ admitting a nearby codeword polynomial $f_z$, the degree bounds also imply $B(X,z)=A(X,z)f_z(X)$. Then the Polishchuk--Spielman (PS) lemma establishes $A(X,Z)\mid B(X,Z)$, with quotient $\sum_{j=0}^{\ell}v_j(X)Z^j$. By unique decoding, every nearby $f_z$ equals $\sum_{j=0}^{\ell}z^jv_j$.

%Finally, a coordinatewise counting argument converts this common structure into correlated agreement. A coordinate $i\in [n]$ such that $(u_{0,i},\dots,u_{\ell,i})\neq (v_{0,i},\dots,v_{\ell,i})$ can become an agreement coordinate for at most $\ell$ challenge. Thus, if $|S|$ challenges are $e$-close, the number of jointly disagreeing coordinates is at most $|S|e/(|S|-\ell)$, giving multiplicative slack $\varepsilon$ when $|S|\ge\ell\left(1+\frac1\varepsilon\right)$.

Extending this argument to AG codes presents two obstacles.
\begin{itemize}
    \item First, for curves of positive genus, analyzing simultaneous Berlekamp–Welch interpolation using only divisor degrees restricts the radius to $e<(n-m)/2-\g$. Thus, to reach half the designed distance, one needs a way to make a Riemann-Roch space vanish without forcing the corresponding divisor to have negative degree.
    \item Second, the coefficients of $A=\sum_{i=0}^{D_Z-1}a_iZ^i$ and $B=\sum_{i=0}^{D_Z+\ell-1}b_iZ^i$ are now functions on $\mathcal X$ with prescribed pole bounds, so the PS lemma for bivariate polynomials is no longer directly applicable. Following the proof of RS codes would therefore require a corresponding divisibility statement over the function field.
\end{itemize}
%\paragraph{Obstacles for one-point AG codes.} 
%polynomial degree constraints of $\deg_X(A)$ and $\deg_X(B)$ in the RS case become pole order bounds in the AG case. A direct degree counting argument in the simultaneous Berlekamp--Welch interpolation loses a genus term. More precisely, if the relevant Riemann--Roch spaces are controlled only through their degrees, the interpolation inequality requires $n-m-2e-2\g>0,$ and hence only gives $e<\frac{n-m}{2}-\g.$  a direct dimension count of relevant Riemann--Roch spaces introduces a genus dependent loss in the unique decoding radius. More precisely, if the relevant Riemann--Roch spaces are controlled only through their degrees, the existence of nonzero solutions of BW equations requires $n-m-2e-2\g>0,$ and hence only gives $e<\frac{n-m}{2}-\g.$ Thus, to reach half the designed distance, one needs a way to make a Riemann--Roch space vanish without forcing the corresponding divisor to have negative degree.

\paragraph{Overcoming the genus penalty.}
%We use a classical auxiliary divisor technique from unique decoding of AG codes. 
From \cite[Proposition~4.2]{Pel96}, for a target error correction bound $t$, one chooses a finite extension $K/\mathbb F_q$ and an auxiliary divisor $F$ of degree $t+\g$ such that $\mathcal L_K(F-Q)=0$ for every effective divisor $0\le Q\le D$ of degree $t+1$. Here, $\mathcal L_K$ denotes a Riemann-Roch space over the extension field $K$. Notice that $\deg(F-Q)=\g-1$, so the vanishing no longer follows from a negative degree argument. %Instead, $F$ is chosen so that all the spaces $\mathcal L_K(F-Q)$ vanish simultaneously. %This is what removes the genus loss and allows decoding up to half the designed distance.

We use the same auxiliary divisor idea in the BW equations to remove the genus loss below half the designed distance. More precisely, we choose a single effective divisor $H$ of degree $d-e+\g-1$ over a finite extension $K/\mathbb F_q$, with support disjoint from those of $D$ and $G$, such that
\[  \mathcal L_K(G+H-D_I)=0,\   \text{for every } I\subseteq[n]\text{ with }|I|=n-e, \]
where $D_I:=\sum_{i\in I}P_i$. Here $\deg(G+H-D_I)=\g-1$ need not be negative, but the choice of $H$ still ensures that any function in $\mathcal L_K(G+H)$ vanishing at $n-e$ evaluation points is zero. We construct $H$ by adding points over $\overline{\mathbb F}_q$, at each step avoiding the finitely many points that would make one of these spaces nonzero. The chosen points then lie in a common finite extension (see Lemmas~\ref{lem:finite-bad-point} and~\ref{lem:auxiliary-divisor}). Thus, the same $H$ works simultaneously for every possible agreement set, independently of $z$.

\paragraph{Divisibility without the PS lemma.}
%We apply Berlekamp--Welch interpolation to the entire received curve, treating $Z$ as an indeterminate. 
We begin with the interpolation setup, which follows the
strategy of~\cite{BCHKS25}, and then explain how error incidence
counting replaces the PS divisibility lemma. As in the RS case,
we enlarge the coefficient spaces to obtain simultaneous
Berlekamp--Welch interpolants of low $Z$-degree but with an auxiliary divisor $H$ described above.

Let $r_i(Z):=\sum_{j=0}^{\ell}u_{j,i}Z^j$ for each coordinate $i\in [n]$. We seek $A(\cdot, Z)=\sum_{j=0}^{D_Z-1}A_jZ^j$ and $B(\cdot,Z)=\sum_{j=0}^{D_Z+\ell-1}B_jZ^j$ with $A_j\in \mathcal L_K(H)$ and $B_j\in \mathcal L_K(G+H)$, satisfying
\begin{equation}\label{eq:bwag}
    A(P_i,Z)r_i(Z)=B(P_i,Z)\ \text{for\ every}\ i\in[n].
\end{equation}
these identities
give a homogeneous linear system over $K$.
Writing $l_K(E):=\dim_K\mathcal L_K(E)$, the system has
$D_Zl_K(H)+(D_Z+\ell)l_K(G+H)$ unknowns and at most
$n(D_Z+\ell)$ equations. By the Riemann--Roch theorem~\cite[Theorem 1.5.15]{Sti09},
$l_K(H)\ge d-e$ and $l_K(G+H)\ge n-e$.
Thus a nonzero solution exists whenever $(d-2e)D_Z>\ell e.$ Lemma~\ref{lem:slack-bw-interpolant} further ensures that such a solution has $A\ne0$.
%This gives more interpolation unknowns and allows us to keep the degrees of $A$ and $B$ in $Z$ small. Let $l_K$ denote the $K$-dimension of the corresponding Riemann--Roch space $\mathcal L_K$. By the Riemann--Roch theorem, $l_K(H)\ge d-e$ and $l_K(G+H)\ge n-e$. To ensure the BW equation has nonzero solutions, it suffices to choose $D_Z$ with \LSSnote{$(d-2e)D_Z>\ell e$}. This gives a solution $(A(\cdot,Z),B(\cdot,Z))$ of Eq.~\eqref{eq:bwag} with $A\neq0$ (see Lemma~\ref{lem:slack-bw-interpolant}).

For every $z\in S$, the difference $B(\cdot,z)-A(\cdot,z)f_z$ belongs to $\mathcal L_K(G+H)$ and vanishes at $n-e$ evaluation points. The choice of $H$ therefore gives $B(\cdot,z)=A(\cdot,z)f_z$ (see Lemma~\ref{lem:vertical-specialization}).
%This interpolation works throughout $e<d/2$, with no additional genus term.

%can we summarize this part as a lemma or some key theory.
\textit{Our departure from the RS argument is to reconstruct a
candidate quotient from a few nearby codeword functions
and prove the global identity directly.} Define 
\[T:=\{i\in[n]: A(P_i,Z)\equiv 0\},\ E_z:=\{i\notin T:r_i(z)\ne f_z(P_i)\}.\]
For $i\notin T$, if a parameter $z\in S$ has an error at coordinate $i$, namely $r_i(z)\neq f_z(P_i)$, then the interpolation equations imply $A(P_i,z)=0$. Since $A(P_i,Z)$ is nonzero and has degree at most $D_Z-1$, at most $D_Z-1$ parameters can have an error at a fixed coordinate. 
Summing over all coordinates outside $T$, the total number of errors is therefore $\sum_{z\in S}|E_z| \le (n-|T|)(D_Z-1) \le n(D_Z-1)$ (see Lemma~\ref{lem:incidence-bound}). If $|S|\ge\ell+1$ and $|S|>(\ell+1)n(D_Z-1)/(e+1)$, averaging gives $\ell+1$ distinct parameters $z_0,\ldots,z_\ell\in S$ with
$\sum_{h=0}^{\ell}|E_{z_h}|\le e$. In particular, their error sets have union $E:=\bigcup_{h=0}^{\ell}E_{z_h}$ of size at most $e$.
%Choose these $\ell+1$ parameters and define $v_0,\ldots,v_\ell$ by Lagrange interpolation:

Define $v_0,\ldots,v_\ell$ as follows
\[\sum_{j=0}^{\ell}v_jZ^j:=\sum_{h=0}^{\ell}f_{z_h}\prod_{\substack{0\le k\le\ell\\ k\ne h}}\frac{Z-z_k}{z_h-z_k}.
\]
Then $v_0,\ldots,v_\ell\in \mathcal{L}(G)$ and $f_{z_h}=\sum_{j=0}^{\ell}v_jz_h^j$ for $0\le h\le\ell$. For each coordinate $i$ outside $T\cup E$, the two polynomials $r_i(Z)$ and $\sum_{j=0}^{\ell}v_j(P_i)Z^j$ have degree at most $\ell$ and agree at $z_0,\ldots,z_\ell$, they coincide identically. The interpolation equations therefore
give $B(P_i,Z)=A(P_i,Z)\sum_{j=0}^{\ell}v_j(P_i)Z^j.$ This identity also holds for $i\in T$, since both $A(P_i,Z)$
and $B(P_i,Z)$ are zero. Consequently, every coefficient of $ B(\cdot,Z)-A(\cdot,Z)\sum_{j=0}^{\ell}v_jZ^j$ lies in $\mathcal L_K(G+H)$ and vanishes at the same set of at least $n-e$ evaluation points. The choice of $H$ therefore gives $B(\cdot,Z)=A(\cdot,Z)\sum_{j=0}^{\ell}v_jZ^j$ and $r_i(Z)=\sum_{j=0}^{\ell}v_j(P_i)Z^j$ for every $i\notin T$.

%they are equal. On $T$, both $A(P_i,Z)$ and $B(P_i,Z)$ are zero. Consequently, every coefficient of $B(\cdot,Z)-A(\cdot,Z)\sum_{j=0}^{\ell}v_jZ^j$ belongs to $\mathcal L_K(G+H)$ and vanishes at $\geq n-e$ evaluation points. The choice of $H$ gives $B(\cdot,Z)=A(\cdot,Z)\sum_{j=0}^{\ell}v_jZ^j$, and hence $r_i(Z)=\sum_{j=0}^{\ell}v_j(P_i)Z^j$ for every $i\notin T$.

For every $z\in S$, to identify every decoded function $f_z$ with $\sum_{j=0}^{\ell}v_jz^j$, including those at parameters
where $A(\cdot,z)\equiv0$, fix a nonzero coefficient $A_j$ of
$A(\cdot,Z)$, the function $A_j(f_z-\sum_{h=0}^{\ell}v_hz^h)\in\mathcal L_K(G+H)$ vanishes wherever $f_z$ agrees with $\mathbf r(z)$. Indeed, $A_j(P_i)=0$ for $i\in T$ and $r_i(z)=\sum_{j=0}^{\ell}v_j(P_i)z^j$ for $i\notin T$. It therefore vanishes at $\geq n-e$ evaluation points and is identically zero by the vanishing property of $H$. Since $A_j\neq0$, we obtain $f_z=\sum_{j=0}^{\ell}v_jz^j$ for every $z\in S$, including when $A(\cdot,z)=0$ (see Lemma~\ref{lem:clean-collinearity}). Thus, without an AG analogue of the PS lemma, we obtain the common polynomial expression needed for correlated agreement.

With this representation established, we conclude by a
coordinatewise counting argument, as in~\cite{goyal2026optimal}. 
At any coordinate where the coefficient tuples
$(u_{0,i},\ldots,u_{\ell,i})$ and
$(v_0(P_i),\ldots,v_\ell(P_i))$ differ, the nonzero polynomial
$r_i(Z)-\sum_{j=0}^{\ell}v_j(P_i)Z^j$ has at most $\ell$ roots.
Hence whenever $|S|\ge\ell+1$, this coordinate contributes an error for at least $|S|-\ell$ parameters in $S$. Since each parameter contributes at most $e$ errors, we obtain the interleaved distance (see Lemma~\ref{lem:collinearity-combinatorial})
\[
    d([\mathbf u_0,\ldots,\mathbf u_\ell],[v_0(D),\ldots,v_\ell(D)])\le \frac{|S|}{|S|-\ell}e.
\]
Together with Lemma~\ref{lem:clean-collinearity}, this proves Theorem~\ref{thm:intro-unique}. For fixed $\ell$ and positive slack, applying Theorem~\ref{thm:ag-corrlated-agreement} with $a=\ell+\lceil\ell/\varepsilon\rceil$ gives $(\ell,e,\sigma,\varepsilon)$-CA with $\sigma=O(1/q)$ whenever $n-m-2e=\Omega(n)$ (see Corollary~\ref{cor:constant}). Taking $a=\ell(e+1)+1$ gives $(\ell,e,\sigma,0)$-CA with the soundness bound in Theorem~\ref{thm:intro-unique} (see Corollary~\ref{cor:without-slack}).

The extension $K$ is used only to construct the auxiliary divisor and the interpolants. All challenges lie in $\F_q$, and the coefficients $v_0,\ldots,v_\ell$ belong to $\mathcal L(G)$ over $\F_q$. Consequently, the final codewords remain over $\F_q$, and the soundness denominator remains $q$.

\subsubsection{Mutual correlated agreement up to the Johnson bound.}
\label{subsec:intro-johnson-technique}

To prove $(\ell,e,\sigma)$-MCA, we bound the set $S$ of parameters for which the required agreement fails. More precisely, $z\in S$ if there exist $f_z\in\mathcal L(G)$ and $A_z\subseteq[n]$ with $|A_z|\ge n-e$ such that $f_z(P_i)=r_i(z)$ for every $i\in A_z$, but no $v_0,\ldots,v_\ell\in\mathcal L(G)$ satisfy $v_j(P_i)=u_{j,i}$ for all $i\in A_z$ and $0\le j\le\ell$. For each $z\in S$, fix one such pair $(f_z,A_z)$. Our goal is to show that $|S|\le \sigma q$.

\paragraph{A common interpolation equation.}
We first seek one polynomial equation satisfied by all the functions $f_z$. Following the Guruswami-Sudan interpolation approach for Reed-Solomon codes~\cite{BCIKS23,BCHKS25}, we keep $Z$ as an indeterminate and construct $Q(Y,Z)\in\mathbb F_q(\mathcal X)[Y,Z]$ for the entire curve. The multiplicity conditions give
\[
Q(f_z,z)=0
\qquad\text{for every }z\in S.
\]
For $e\le(1-\sqrt{\rho}-\eta)n$ and fixed $\rho\in(0,1)$, the same low $Z$-degree counting idea as in~\cite[Section~3.1]{BCHKS25} gives $\deg_Y Q\le D_Y=O(\eta^{-1})$ and $\deg_Z Q\le D_Z=O(\ell\eta^{-2})$ (see Lemma~\ref{lem:ag-uniform-interpolation}).

\paragraph{Overcoming the difficulty of infinite local expansions.}
For RS codes, a message is a polynomial $f(X)$ of degree at most $m$. Root lifting gives the coefficients $a_0,\ldots,a_m$ of its Taylor expansion at a point $x_0$. Since $f(X)=\sum_{j=0}^{m}a_j(X-x_0)^j,$ we can recover every evaluation $f(x_i)$ as a fixed linear combination of these coefficients. For AG codes, however, the local expansion of a function is generally infinite, so a finite truncation cannot be used in the same way to evaluate the function at other points.

We instead use the pole bound to recover each evaluation as a fixed linear combination of finitely many local coefficients. Choose an auxiliary point $P_0\in\mathcal X(\overline{\mathbb F}_q)$ outside $P_\infty,P_1,\ldots,P_n$, and let $t$ be a local parameter at $P_0$. For $h\in\mathcal L_{\overline{\mathbb F}_q}(aP_\infty)$, the first $a+1$ local coefficients uniquely determine $h$: if they all vanish, then $h$ has a zero of order at least $a+1$ but poles of total degree at most $a$, so $h=0$. By linearity, each value $h(P_i)$ is therefore a fixed linear combination of these local coefficients. Thus, for each $i\in[n]$, there exist constants $\lambda_{i,0},\ldots,\lambda_{i,a}$, independent of $h$, such that
\[
h(P_i)=\sum_{j=0}^{a}\lambda_{i,j}[t^j]h
\qquad
\text{for every }
h\in\mathcal L_{\overline{\mathbb F}_q}(aP_\infty).
\]
%This is Lemma~\ref{lem:ag-local-recovery}.

To apply this formula while retaining the degree bounds needed for soundness, we combine it with the RS root-lifting argument. We first factor $Q$ over $\overline{\mathbb F}_q(\mathcal X)$ and write $Q(Y,Z)=C(Z)\prod_{\nu=1}^{r}R_\nu(Y^{\mu_\nu},Z)^{e_\nu},$ where $C(Z)$ is independent of $Y$, $e_\nu\ge 1$ and $\mu_\nu$ is a power of the characteristic chosen so that $\partial_Y R_\nu\ne 0$. Here $\nu$ indexes the irreducible factors of positive $Y$-degree (see Section~\ref{subsec:ag-factorization}). Choose $P_0$ also to satisfy Lemma~\ref{lem:ag-common-auxiliary-point}. After the exclusions in Section~\ref{subsec:ag-factorization}, $Q(f_z,z)=0$ implies that $R_\nu(f_z^{\mu_\nu},z)=0$ for some $\nu$. We assign each remaining parameter $z$ to one such factor.

Root lifting expresses the local coefficients of $f_z^{\mu_\nu}$ as rational functions of $y_z:=f_z(P_0)^{\mu_\nu}$ and $z$. Since $f_z^{\mu_\nu}\in \mathcal L_{\overline{\mathbb F}_q}(\mu_\nu mP_\infty)$, the recovery formula with $a=\mu_\nu m$ gives $f_z(P_i)^{\mu_\nu}$. After clearing denominators, we obtain a polynomial $F_i(Y,Z)$ for each coordinate $i$ such that
\[
F_i(y_z,z)=0
\quad\Longleftrightarrow\quad
f_z(P_i)=r_i(z).
\]
The degrees of these polynomials satisfy the bounds needed for the subsequent counting argument (see Lemmas~\ref{lem:ag-hensel-lifting} and~\ref{lem:ag-coordinate-polynomials}).

\paragraph{Completing the proof as in the RS case.}
It remains to bound $|S|$ using these coordinate equations. As in the RS proof, we use factorization to partition the remaining parameters into at most $D_Y$ sets. For each set and coordinate $i$, the resultant argument shows that either all parameters satisfy $f_z(P_i)=r_i(z)$, or at most a bounded number do. This bound does not depend on the size of the set (Theorem~\ref{thm:ag-parameter-grouping}).

Suppose a set has at most $m$ common agreement coordinates. Each parameter must then agree at at least $n-e-m>0$ other coordinates. However, at each of these other coordinates, only a bounded number of parameters can agree. Counting these agreements gives an upper bound on the size of the set. Thus, any sufficiently large set must have more than $m$ common agreement coordinates. Lagrange interpolation and the bound on the number of zeros then show that $f_z=\sum_{j=0}^{\ell}v_jz^j$ for $v_j\in\mathcal L(G),$ throughout that set. The same coordinatewise counting as in the RS proof bounds its size by $\ell(e+1)$, since all these parameters are MCA-bad (Lemmas~\ref{lem:ag-large-group} and~\ref{lem:ag-common-to-correlated}).

The degree bounds also control the smaller sets and the excluded parameters. Summing over all sets gives $|S|=O\!\left(\frac{\ell n}{\eta^3}\right).$ Hence $\mathcal C$ has $(\ell,e,\sigma)$-MCA with $\sigma=O\!\left(\frac{\ell n}{\eta^3q}\right)$ whenever $e\le(1-\sqrt{\rho}-\eta)n$, for fixed $\rho\in(0,1)$, proving Theorem~\ref{thm:ag-johnson}.

\subsubsection{Counterexamples for correlated agreement.}
\label{subsec:intro-counterexample-technique}

Our two constructions are both based on the Garcia-Stichtenoth tower.
The common ingredient is a function $x$ with
$(x)_\infty=bP_\infty$ such that each nonzero field element
occurs at exactly $b$ evaluation points.
Substituting $x$ into a polynomial therefore converts its
nonzero roots into an exact agreement count.
In both constructions, two leading monomials define an
affine line, while the remaining terms yield nearby
codewords. A pole order bound then rules out correlated
agreement for the full AG code.

\paragraph{Trace-polynomial construction.}
We adapt the RS counterexample of~\cite{DEEP20}, recalled
in~\cite{BCHKS25}. Write $s=p^a$ and $q=s^{2h}$, and let
$\mathrm{Tr}_{q/s}:\mathbb F_q\to\mathbb F_s$ be the trace function. For each $\beta\in\mathbb F_q^*$, the polynomial $\mathrm{Tr}(\beta Y)^{s-1}-1$ vanishes precisely where $\mathrm{Tr}(\beta Y)\ne0$, giving $q-q/s$ distinct nonzero roots. After normalization, its two leading terms define $\mathbf u_0+z\mathbf u_1$, while the remaining terms, after substitution $Y=x$, lie in $\mathcal L(G)$.
As $\beta$ varies, we obtain $(q-1)/(s-1)$ distinct
challenges whose corresponding points have distance
$e=n-(q-q/s)b$ from the code.

In contrast, the pole of the function defining
$\mathbf u_0$ cannot be canceled by any function in
$\mathcal L(G)$. Zero counting gives 
\[
d([\mathbf u_0,\mathbf u_1],\mathcal C^2)
\ge (2-1/s)e,
\]
which rules out CA with slack $\varepsilon<1-1/s$
unless $\sigma\ge(q-1)/((s-1)q)$. Increasing $s$ brings the radius closer to half the minimum distance, but changes the construction parameters. For fixed $p,a$, increasing $h$ makes $q$ arbitrarily large, without reducing the constant soundness lower bound.

\paragraph{Subspace-annihilator construction.}
The second construction adapts the Gabidulin counterexample
of~\cite[Section~3.3]{LXYZ26}. Fix $q=Q^6$. For every $z\in\mathbb F_q$, there exists a two-dimensional $\mathbb F_Q$-subspace $U$ with annihilator $A_U(Y)=Y^{Q^2}-zY^Q+\beta Y.$
The key modification is to use the complementary roots:
\[
F_U(Y):=\frac{Y^q-Y}{A_U(Y)}
=\prod_{c\in\mathbb F_q\setminus U}(Y-c).
\]
Polynomial division isolates the parameter $z$ in the
second leading term:
\[
F_U(Y)=Y^{q-Q^2}
+zY^{q-2Q^2+Q}+H_U(Y),\ 
\deg H_U\le q-2Q^2+1.
\]
Choosing $m=(q-2Q^2+1)b$ places $-H_U(x)$ in
$\mathcal L(G)$. Since $F_U$ has $q-Q^2$ nonzero roots,
every point on the resulting affine line has distance
$e=n-(q-Q^2)b$ from the code.
Nevertheless, the pole of the second leading monomial gives
\[
d([\mathbf u_0,\mathbf u_1],\mathcal C^2)
\ge e+(Q^2-Q)b.
\]
Thus all challenges admit nearby codewords, but no common
pair of codewords satisfies the required agreement,
ruling out every $\sigma<1$ for the admissible slack.
Reversing the two directions also gives $q-1$ close points
and one point beyond the slacked radius, yielding the
proximity-gap lower bound $\sigma\ge1-1/q$.

Unlike the trace construction, this construction approaches
half the minimum distance by increasing only the tower
level. Indeed, writing $t_j=n_j-(q-1)b_j$, we have
$t_j>0$, $t_j/b_j\to0$ and $\frac{e_j}{n_j}-\frac{d(\mathcal C_j)}{2n_j}
=\frac{t_j}{2n_j}\to 0.$ For every fixed $\varepsilon<Q/(Q+1)$, sufficiently high
levels remain below Johnson and retain both soundness
lower bounds, while the field stays fixed and the rate
and relative error radius converge to positive constants.

\section{Preliminaries}\label{sec:preliminaries}
\subsection{Notations}
\label{subsec:basic-notation}

For a positive integer $n$, write $[n]=\{1,\ldots,n\}$. Let $\mathbb F_q$ be the finite field of order $q$, let $p=\mathrm{char}\mathbb F_q$, and let $\overline{\mathbb F}_q$ be an algebraic closure of $\mathbb F_q$. We use $u_i$ to denote the $i$-th coordinate of $\mathbf u$. For $A\subseteq[n]$, write $\mathbf u|_A:=(u_i)_{i\in A}$. For a code $\mathcal C\subseteq\mathbb F_q^n$, its restriction to $A$ is $\mathcal C|_A=\{\mathbf c|_A:\mathbf c\in\mathcal C\}$.

For $\mathbf{u},\mathbf{v}\in \F_q^n$, the Hamming distance is defined as $d(\mathbf u,\mathbf v)=|\{i\in[n]:u_i\neq v_i\}|$. For a nonempty code $\mathcal C$, let $d(\mathbf u,\mathcal C)=\min_{\mathbf c\in\mathcal C}d(\mathbf u,\mathbf c)$. Its minimum distance is $d(\mathcal C)=\min_{\mathbf c\neq\mathbf c'\in\mathcal C} d(\mathbf c,\mathbf c')$. The corresponding relative distances are obtained by dividing by $n$.

For an integer $\ell\ge1$ and $\mathbf u_0,\ldots,\mathbf u_\ell\in\mathbb F_q^n$, write $[\mathbf u_0,\ldots,\mathbf u_\ell]=((u_{0,i},\ldots,u_{\ell,i}))_{i\in[n]}$ for their interleaving and define $\mathcal C^{\ell+1}=\{[\mathbf c_0,\ldots,\mathbf c_\ell]:\mathbf c_0,\ldots,\mathbf c_\ell\in\mathcal C\}$. We view $[\mathbf{u}_0,\ldots,\mathbf{u}_\ell]$ and $[\mathbf{c}_0,\ldots,\mathbf{c}_\ell]$ as words in $(\mathbb F_q^{\ell+1})^n$ and define
\[
    d([\mathbf u_0,\ldots,\mathbf u_\ell],[\mathbf c_0,\ldots,\mathbf c_\ell])= \bigl|\{i\in[n]:(u_{0,i},\ldots,u_{\ell,i}) \neq(c_{0,i},\ldots,c_{\ell,i})\}\bigr|.
\]
In particular, $d([\mathbf u_0,\ldots,\mathbf u_\ell],\mathcal C^{\ell+1})\leq e$ means that there exist $\mathbf c_0,\ldots,\mathbf c_\ell\in\mathcal C$ that agree with $\mathbf u_0,\ldots,\mathbf u_\ell$, respectively, on a common set of at least $n-e$ coordinates.

\subsection{Proximity Gaps and Agreement Notions}
\label{subsec:agreement-definitions}

We now define proximity gaps (PG), correlated agreement (CA) and mutual correlated agreement (MCA) for polynomial curves of degree at most $\ell$, where $\ell\ge1$ is an integer. Fix a linear code $\mathcal C\subseteq\mathbb F_q^n$, an integer $0\leq e\leq n$, a soundness parameter $0\leq\sigma\leq1$ and a multiplicative slack $\varepsilon\ge0$. For $\mathbf u_0,\ldots,\mathbf u_\ell\in\mathbb F_q^n$, write $\mathbf r(Z)=\sum_{j=0}^{\ell}\mathbf u_j Z^j$ and $r_i(Z)=\sum_{j=0}^{\ell}u_{j,i}Z^j$ where $Z$ is an indeterminate. The case $\ell=1$ gives affine lines. These curves lie in $\mathbb F_q^n$ and are distinct from the algebraic curve used to define an AG code. All probabilities over $z\in\mathbb F_q$ use the uniform distribution on parameters and we do not require $\mathbf r$ to be injective.

\begin{defn}[Proximity gaps]
\label{def:proximity-gaps}
We say that $\mathcal C$ has \emph{$(\ell,e,\sigma,\varepsilon)$-PG} if, for every $\mathbf u_0,\ldots,\mathbf u_\ell\in\mathbb F_q^n$, the inequality $\Pr_z[d(\mathbf r(z),\mathcal C)\leq e]>\sigma$ implies that $d(\mathbf r(z),\mathcal C)\leq(1+\varepsilon)e$ for every $z\in\mathbb F_q$. Thus either every point on the curve is within distance $(1+\varepsilon)e$ of the code, or at most a $\sigma$ fraction of the parameters satisfy $d(\mathbf r(z),\mathcal C)\leq e$.
\end{defn}

\begin{defn}[Correlated agreement]
\label{def:correlated-agreement}
We say that $\mathcal C$ has \emph{$(\ell,e,\sigma,\varepsilon)$-CA} if, for every $\mathbf u_0,\ldots,\mathbf u_\ell\in\mathbb F_q^n$, the inequality $\Pr_z[d(\mathbf r(z),\mathcal C)\leq e]>\sigma$ implies that $d([\mathbf u_0,\ldots,\mathbf u_\ell],\mathcal C^{\ell+1})\leq(1+\varepsilon)e$.
\end{defn}

We follow the definitions in~\cite{BCIKS23,goyal2026optimal}, with curves of degree at most $\ell$ at a fixed radius $e$. A property holds up to radius $e$ if it holds at every integer radius from $0$ to $e$.

\begin{rmk}[Slack]
\label{rmk:agreement-slack}
The parameters $(\ell,e,\sigma,\varepsilon)$ specify the degree bound, the integer error radius, the soundness bound and the multiplicative slack, respectively. Thus the radius is $e$, while the conclusion allows radius $(1+\varepsilon)e$. The case $\varepsilon=0$ is called \emph{without slack}. We use absolute radii $e$; the corresponding relative radius is $e/n$.
\end{rmk}

For mutual correlated agreement, we use the formulation in terms of ``specified agreement sets'' from~\cite[Section~4.3, Theorem~4.6]{BCHKS25}. It is also called \emph{list correlated agreement} there.

\begin{defn}[Mutual correlated agreement]
\label{def:mutual-correlated-agreement}
We say that $\mathcal C$ has \emph{$(\ell,e,\sigma)$-MCA} if, for every $\mathbf u_0,\ldots,\mathbf u_\ell\in\mathbb F_q^n$,
\[
    \Pr_{z\in \F_q}\!\left[
    \begin{array}{l}
        \exists A\subseteq[n],\ |A|\geq n-e:\\
        \mathbf r(z)|_A\in\mathcal C|_A
        \quad\text{and}\quad
        [\mathbf u_0,\ldots,\mathbf u_\ell]|_A\notin\mathcal C^{\ell+1}|_A
    \end{array}
    \right]\leq\sigma.
\]
\end{defn}

Equivalently, an \emph{$e$ agreement witness} is a triple $(z,\mathbf c,A)$ with $z\in\mathbb F_q$, $\mathbf c\in\mathcal C$, $A\subseteq[n]$, $|A|\geq n-e$ and $\mathbf c|_A=\mathbf r(z)|_A$. It is \emph{MCA-bad} if no $\mathbf c_0,\ldots,\mathbf c_\ell\in\mathcal C$ satisfy $\mathbf c_j|_A=\mathbf u_j|_A$ for every $0\le j\le\ell$.
Thus $(\ell,e,\sigma)$-MCA bounds the probability that $z$ has at least one \emph{MCA-bad} witness. Whenever a parameter $z$ has no \emph{MCA-bad} witness, every agreement witness $(z,\mathbf c,A)$ admits $\mathbf c_0,\ldots,\mathbf c_\ell\in\mathcal C$ such that $\mathbf c_j|_A=\mathbf u_j|_A$ for every $0\le j\le\ell$. We do not define MCA with slack, so its parameters are $(\ell,e,\sigma)$.

\begin{rmk}\label{rmk:mca-ca-pg}
    It is well known that $(\ell,e,\sigma)$-MCA implies $(\ell,e,\sigma,0)$-CA, which in turn implies $(\ell,e,\sigma,0)$-PG. More generally, $(\ell,e,\sigma,\varepsilon)$-CA implies $(\ell,e,\sigma,\varepsilon)$-PG. We omit the proof. In each setting, we prove the stronger property and the corresponding weaker properties then follow immediately.
\end{rmk}

\subsection{One-Point Algebraic Geometric Codes}
\label{subsec:ag-codes}

Let $\mathcal X$ be a smooth, projective, geometrically irreducible curve over $\mathbb F_q$ of genus $\g$, with function field $\mathbb F_q(\mathcal X)$. Write $\mathcal X(\mathbb F_q)$ for its set of rational points. For a divisor $A$ on $\mathcal X$, define
\[
    \mathcal L(A)=\{0\}\cup
    \{f\in\mathbb F_q(\mathcal X)^*:(f)+A\geq0\},
    \qquad l(A)=\dim_{\mathbb F_q}\mathcal L(A),
\]
where $(f)$ denotes the principal divisor of $f$. The Riemann-Roch theorem gives $l(A)\geq\deg A+1-\g$, with equality when $\deg A\geq2\g-1$. %If $\deg A<0$, then $\mathcal L(A)=\{0\}$.

Let $P_1,\ldots,P_n,P_\infty\in\mathcal X(\mathbb F_q)$ be pairwise distinct, and set $D=P_1+\cdots+P_n$ and $G=mP_\infty$ where $m<n$ is a positive integer. The corresponding one-point algebraic geometric code is
\[
    \mathcal C=\cC_{\mathcal L}(D,G)
    =\{f(D):f\in\mathcal L(G)\}\subseteq\mathbb F_q^n,
    \qquad f(D):=(f(P_1),\ldots,f(P_n)).
\]
The evaluation map is injective, so $\dim_{\mathbb F_q}\mathcal C=l(G)\geq m+1-\g$ with equality when $m\geq2\g-1$. Moreover, $d(\mathcal C)\geq n-m$ where $d:=n-m$ is the designed minimum distance. Thus every received word has at most one codeword within distance $e$ when $0\leq e<(n-m)/2$. We write $\rho:=m/n$. Note that $\rho$ is determined by the degree of $G$, whereas the actual rate of $\mathcal C$ is $l(G)/n$, which equals $(m+1-\g)/n$ when $m\geq 2\g-1$.

For this evaluation code, we write $(z,f,A)$ for the $e$ agreement witness $(z,f(D),A)$ from Definition~\ref{def:mutual-correlated-agreement} where $f\in\mathcal L(G)$. We call this witness \emph{MCA-bad} if there are no $v_0,\ldots,v_\ell\in\mathcal L(G)$ such that $v_j(P_i)=u_{j,i}$ for every $i\in A$ and $0\le j\le\ell$.

We also fix the local notation used later. For $P\in\mathcal X(\mathbb F_q)$, let $v_P$ denote the normalized discrete valuation at $P$, with $v_P(0)=+\infty$ and let $\mathcal O_P=\{f\in\mathbb F_q(\mathcal X):v_P(f)\geq0\}$ be its local ring. For a local parameter $t$ at $P$, every $f\in\mathcal O_P$ has a unique formal expansion $f=\sum_{j\geq0}a_jt^j$ with $a_j\in\mathbb F_q$. We write $[t^j]f=a_j$ and in particular, $[t^0]f=f(P)$.

%For any field extension $L/\mathbb F_q$, let $\mathcal X_L:=\mathcal X\times_{\mathbb F_q}L$ be the curve obtained by extending the constant field. Its function field is denoted by $L(\mathcal X)$, and $\mathcal L_L(A)$ denotes the Riemann-Roch space of a divisor $A$ on $\mathcal X_L$. We retain the notation $P_\infty,P_1,\ldots,P_n$ for the corresponding points on $\mathcal X_L$, and use the same notation for divisors after base change. The same local notation is used after extending the constant field, with expansion coefficients in $L$. In what follows, $L$ will typically be either a finite extension of $\mathbb F_q$ or the algebraic closure $\overline{\mathbb F}_q$.

For any field extension $K/\mathbb F_q$, let $\mathcal X_{K}=\mathcal X\times_{\mathbb F_q}K$ be the curve obtained by extending the constant field. Its function field is denoted by $K(\mathcal X)$ and $\mathcal L_{K}(A)$ denotes the Riemann-Roch space of a divisor $A$ on this curve over $K$. We retain the notation $P_\infty,P_1,\ldots,P_n$ for the corresponding points on $\mathcal X_{K}$. The same local notation is used after extending the constant field, with expansion coefficients in $K$.

\subsection{The Garcia-Stichtenoth tower}\label{sec:GS-tower}
In this subsection, let $q=s^2$, where $s$ is a prime power, and let $x_0$ be an indeterminate over $\mathbb F_q$. We use the Garcia-Stichtenoth tower~\cite{GarciaStichtenoth95}, defined recursively by
\[
F_0=\mathbb F_q(x_0),\qquad F_{j+1}=F_j(y_{j+1}),
\]
where $y_{j+1}^{s}+y_{j+1}=x_j^{s+1}$ and $x_{j+1}=\frac{y_{j+1}}{x_j}.$

For each $j\ge 0$, let $\mathcal X_j$ be the curve with function field $F_j$. The function $x_0$ has a unique pole on $\mathcal X_j$, which we denote by $P_{\infty,j}$. We use all the other rational points of $\mathcal X_j$ for evaluation and write
\[
b_j=[F_j:F_0],\qquad
D_j=\sum_{P\in\mathcal X_j(\mathbb F_q)\setminus\{P_{\infty,j}\}}P,
\qquad n_j=\deg D_j.
\]

\begin{lemma}\label{lem:GS-tower}
For the Garcia-Stichtenoth tower defined above, we have $b_j=s^j$, $(x_0)_\infty=b_jP_{\infty,j}$ and $(q-1)b_j<n_j\le qb_j.$

Moreover, for every $a\in\mathbb F_q^*$, exactly $b_j$ evaluation points satisfy $x_0(P)=a$. Writing $\g_j$ for the genus of $\mathcal X_j$, we also have $n_j\to\infty$, $\frac{\g_j}{n_j}\longrightarrow\frac1{s-1}$ and $\frac{n_j}{b_j}\longrightarrow q-1.$
\end{lemma}

\begin{proof}
At level $0$, the function $x_0$ has a simple pole at infinity. Suppose that $x_j$ has valuation $-1$ at $P_{\infty,j}$. Let $Q$ be a place of $F_{j+1}$ lying above $P_{\infty,j}$, with ramification index $a$. The recursive equation implies $v_Q(y_{j+1})<0$ and $s\cdot v_Q(y_{j+1})=-a(s+1).$ Since $\gcd(s,s+1)=1$ and $a\le[F_{j+1}:F_j]\le s$, we must have $a=s$. Thus $P_{\infty,j}$ is totally ramified, with a unique rational place $Q$ above it. Moreover, $v_Q(x_{j+1})=-(s+1)+s=-1.$ By induction, $b_j=s^j$ and $(x_0)_\infty=b_jP_{\infty,j}$.

Next, suppose that $x_j(P)=a\ne0$ at a rational place $P$. Since $a^{s+1}\in\mathbb F_s^*$ and the map $Y\mapsto Y^s+Y$ is the trace map from $\mathbb F_q$ to $\mathbb F_s$, the equation $Y^s+Y=a^{s+1}$ has exactly $s$ distinct nonzero solutions in $\mathbb F_q$. The corresponding values $x_{j+1}=Y/a$ are also nonzero. Hence every place with $x_0=a\in\mathbb F_q^*$ splits completely throughout the tower. In particular, exactly $b_j$ evaluation points satisfy $x_0(P)=a$, giving $n_j\ge(q-1)b_j.$

There is also a rational place at which $x_0,x_1,\ldots,x_j$ all vanish. Indeed, starting with $x_0=0$ at level $0$, the simple root $Y=0$ of the reduced equation $Y^s+Y=0$ gives a rational place at the next level. At this place, $v(y_{j+1})=(s+1)v(x_j)$, so $v(x_{j+1})=s\,v(x_j)>0$. This proves the assertion by induction, and gives $n_j>(q-1)b_j$.

On the other hand, every evaluation point has $x_0(P)\in\mathbb F_q$. For each $a\in\mathbb F_q$, the function $x_0-a$ has pole degree $b_j$, so it has at most $b_j$ zeros among the evaluation points. Summing over all $a\in\mathbb F_q$ gives $n_j\le qb_j$.

Finally, $n_j\to\infty$ follows from $n_j\ge(q-1)s^j$, and the limit $\g_j/n_j\to1/(s-1)$ is the asymptotic optimality result for this tower~\cite{GarciaStichtenoth95}. The genus calculation also gives $\g_j/b_j\to s+1$; see~\cite[Remark~5.10]{GS07}. Consequently, $n_j/b_j=(n_j/\g_j)(\g_j/b_j)\to(s-1)(s+1)=q-1$.
\end{proof}

\subsection{Polynomial Notation and Resultants}
\label{subsec:polynomial-preliminaries}

For a polynomial or formal power series $F$, we write $[Z^j]F$ for the coefficient of $Z^j$ in $F$. We use the same notation for other variables, for example, $[t^aV^b]F$ is the coefficient of $t^aV^b$ in $F$.

For a field $K$ and a polynomial $P\in K[Y,Z]$, we write $\deg_Y P$ and $\deg_Z P$ for its degrees in $Y$ and $Z$, respectively, with both degrees of the zero polynomial defined to be $-\infty$. For $P\neq 0$, we regard $P$ as an element of $K[Z][Y]$ and write $\mathrm{lc}_Y(P)\in K[Z]$ for its leading coefficient in $Y$. The monic greatest common divisor in $K[Z]$ of the coefficients of $P$ is called its \emph{$Y$-content}. We say that $P$ is primitive in $Y$, or \emph{$Y$-primitive}, if its $Y$-content is $1$.

\begin{defn}
\label{def:ag-resultant}
Let $K$ be a field and let $P,T\in K[Y,Z]\setminus\{0\}$. The \emph{Sylvester matrix} $\mathrm{Syl}_Y(P,T)$ is the $(\deg_Y P+\deg_Y T)\times(\deg_Y P+\deg_Y T)$ matrix over $K[Z]$ whose rows are the coefficient vectors of $Y^{\deg_Y T-1}P,\ldots,P$ and $Y^{\deg_Y P-1}T,\ldots,T$ in the indicated order, with columns indexed by the descending powers of $Y$ from $Y^{\deg_Y P+\deg_Y T-1}$ to $1$. A block of rows is omitted when its number of rows is zero, and the determinant of the empty matrix is taken to be $1$. The \emph{resultant} of $P$ and $T$ with respect to $Y$ is $\mathrm{Res}_Y(P,T) :=\det\mathrm{Syl}_Y(P,T)\in K[Z]$.
\end{defn}

We record the following well-known properties of the resultant.

\begin{lemma}
\label{lem:ag-resultant-properties}
Let $K$ be a field and let $P,T\in K[Y,Z]\setminus\{0\}$ with $\deg_Y P>0$. Then $\deg_Z\mathrm{Res}_Y(P,T) \leq (\deg_Y P)(\deg_Z T) +(\deg_Y T)(\deg_Z P)$. If $P$ and $T$ are coprime in $K(Z)[Y]$, then $\mathrm{Res}_Y(P,T)\neq 0$. If $(y,z)\in K^2$ satisfies $P(y,z)=T(y,z)=0$, then $\mathrm{Res}_Y(P,T)(z)=0$.
\end{lemma}

%\section{Correlated Agreement via Berlekamp-Welch Interpolation}\label{sec:ag-unique}
%\input{AG-Proximity-Gaps}

\section{Correlated Agreement up to half the designed distance}
Throughout this section, we use the one-point algebraic geometric code $\mathcal C=\cC_{\mathcal L}(D,G)$ and the notation from Section~\ref{subsec:ag-codes}. Assume that $l(G)=m+1-\g$, as holds when $m\geq2\g-1$. 

For a finite extension $K/\mathbb{F}_q$ and a divisor $A$ on the base-changed curve $\mathcal X_K$, we write $\mathcal L_K(A)$ for the corresponding Riemann--Roch space over $K$ and $l_K(A):=\dim_K\mathcal L_K(A)$. We write $\mathcal L(G)=\mathcal L_{\mathbb F_q}(G)$. We use the same notation for the points $P_i$, the divisor $D$, and the divisor $G$ after base change from $\mathbb{F}_q$ to $K$.

Fix an integer $\ell\ge1$. We use Berlekamp--Welch interpolation to prove correlated agreement for curves of degree at most $\ell$ for $0\leq e<d/2$, with versions with and without slack. %This range has no loss relative to half the designed minimum distance.

\begin{thm}\label{thm:ag-corrlated-agreement}
Let $e$ be an integer with $0\le e<d/2$. Choose an integer $D_Z\ge 1$ such that $(d-2e)D_Z>\ell e.$ For $\mathbf{u}_0,\ldots,\mathbf{u}_\ell\in \mathbb{F}_q^n$, let $S=\{z\in\mathbb{F}_q:d(\mathbf{r}(z),\mathcal{C})\le e\}$. If there is an integer $a\ge\ell+1$ such that
\[
|S|\ge\max\left\{\left\lfloor\frac{(\ell+1)n(D_Z-1)}{e+1}\right\rfloor+1,\ a\right\},
\]
then there exist $v_0,\ldots,v_\ell\in \mathcal{L}(G)$ such that $d([\mathbf{u}_0,\ldots,\mathbf{u}_\ell],[v_0(D),\ldots,v_\ell(D)])\le \frac{a}{a-\ell}e.$ In particular, $d([\mathbf{u}_0,\ldots,\mathbf{u}_\ell],\mathcal{C}^{\ell+1})\le \frac{a}{a-\ell}e.$
\end{thm}

We prove Theorem~\ref{thm:ag-corrlated-agreement} in several steps. Throughout the proof, we fix $e,D_Z,\mathbf{u}_0,\ldots,\mathbf{u}_\ell$ and the set $S$ as above. For each $z\in S$, let $f_z\in \mathcal{L}(G)$ be the unique function satisfying $d(\mathbf{r}(z),f_z(D))\le e.$

We first choose an auxiliary divisor, which may be defined over a finite extension of $\F_q$. This does not change the final conclusion, because the final $\ell+1$ codewords will still come from the original space $\mathcal{L}(G)$.

\begin{lemma}\label{lem:finite-bad-point}
Let $\mathcal X$ be a smooth projective irreducible curve of genus $\g$ over an algebraically closed field $k$. Let $A$ be a divisor on $\mathcal X$ satisfying $l(A)=0$ and $\deg(A)\le \g-2.$ Then there are only finitely many points $P\in \mathcal X(k)$ such that $l(A+P)>0.$ More precisely, the number of such points is at most $2\g-2-\deg(A).$
\end{lemma}

\begin{proof} 
Fix a canonical divisor $K_{\mathcal X}$ on $\mathcal X$. By the Riemann--Roch theorem,
\[
l(K_{\mathcal X}-A)=l(A)-\deg(A)-1+\g=\g-1-\deg(A).
\]
Since $\deg(A)\le \g-2$, we have $l(K_{\mathcal X}-A)\ge 1.$ We may therefore fix a nonzero function $f\in \mathcal L(K_{\mathcal X}-A).$

Now let $P\in \mathcal X(k)$. Applying the Riemann--Roch theorem to $A$ and $A+P$, respectively, gives $l(A)-l(K_{\mathcal X}-A)=\deg(A)+1-\g$ and $l(A+P)-l(K_{\mathcal X}-A-P)=\deg(A)+2-\g.$ Subtracting the first equality from the second, we obtain
\[
l(A+P)-l(A)=1+l(K_{\mathcal X}-A-P)-l(K_{\mathcal X}-A).
\]

Since $\mathcal L(K_{\mathcal X}-A-P)\subseteq \mathcal L(K_{\mathcal X}-A),$ and subtracting one point can decrease the dimension of a Riemann--Roch space by at most one, we have
\[
l(K_{\mathcal X}-A)-1\le l(K_{\mathcal X}-A-P)\le l(K_{\mathcal X}-A).
\]
Consequently, $l(A+P)-l(A)\in\{0,1\}.$

Suppose now that $P$ satisfies $l(A+P)>0$. Since $l(A)=0$, it follows that $l(A+P)=1.$ The equality above then implies $l(K_{\mathcal X}-A-P)=l(K_{\mathcal X}-A).$ Together with the inclusion $\mathcal L(K_{\mathcal X}-A-P)\subseteq \mathcal L(K_{\mathcal X}-A),$ this gives $\mathcal L(K_{\mathcal X}-A-P)=\mathcal L(K_{\mathcal X}-A).$ In particular, the fixed nonzero function $f$ also belongs to $\mathcal L(K_{\mathcal X}-A-P)$. Hence, $\mathrm{div}(f)+K_{\mathcal X}-A-P\ge 0.$ Therefore, $P$ belongs to the support of the effective divisor $E_f:=\mathrm{div}(f)+K_{\mathcal X}-A.$

Thus, every point $P$ satisfying $l(A+P)>0$ lies in $\mathrm{supp}(E_f)$. Since $E_f$ is an effective divisor, its support is finite. Moreover, $\deg(E_f)=\deg(K_{\mathcal X}-A)=2\g-2-\deg(A),$ so the number of such points is at most $2\g-2-\deg(A).$
\end{proof}

\begin{lemma}\label{lem:auxiliary-divisor}
Using the above definitions, fix an integer $0\le e<d$, and set $h:=d-e+\g-1.$ Then there exist a finite extension $K/\mathbb{F}_q$ and an effective divisor $H$ on $\mathcal X_K$ of degree $h$ such that $\mathrm{supp}(H)\cap \bigl(\mathrm{supp}(D)\cup\mathrm{supp}(G)\bigr) =\emptyset$ and for every subset $I\subseteq[n]$ with $|I|=n-e$,
\[
\mathcal L_K\left(G+H-\sum_{i\in I}P_i\right)=\{0\}.
\]
\end{lemma}

\begin{proof}
Let $\overline{\mathbb{F}}_q$ be an algebraic closure of $\mathbb{F}_q$, and consider the base-changed curve $\overline{\mathcal X}:=\mathcal X_{\overline{\mathbb{F}}_q}.$ We use the same notation for the divisors and points after base change to $\overline{\mathbb{F}}_q$.

For every subset $I\subseteq[n]$ with $|I|=n-e$, define $D_I:=\sum_{i\in I}P_i.$ We shall choose pairwise distinct points $Q_1,\ldots,Q_h\in\overline{\mathcal X}(\overline{\mathbb{F}}_q)$ inductively. Set $H_0:=0$ and $A_{I,0}:=G-D_I.$ Since $\deg(A_{I,0})=m-(n-e)=e-d<0,$ we have $l_{\overline{\mathbb{F}}_q}(A_{I,0})=0$ for every $I$. Suppose that, for some $0\le t\le h-1$, we have already chosen $Q_1,\ldots,Q_t$. Define
\[
H_t:=Q_1+\cdots+Q_t,\qquad A_{I,t}:=G+H_t-D_I,
\]
and assume inductively that $l_{\overline{\mathbb{F}}_q}(A_{I,t})=0$ for every $I$.

Since $t\le h-1$, we have $\deg(A_{I,t})=m+t-(n-e)=e-d+t\leq \g-2$. Therefore, by Lemma~\ref{lem:finite-bad-point}, for each fixed $I$, there are only finitely many points $P\in\overline{\mathcal X}(\overline{\mathbb{F}}_q)$ such that $l_{\overline{\mathbb{F}}_q}(A_{I,t}+P)>0.$ There are only finitely many subsets $I\subseteq[n]$ of size $n-e$, so the union of all these bad-point sets is finite.

We additionally require $Q_{t+1}$ to avoid $\mathrm{supp}(D)\cup\mathrm{supp}(G) \cup\{Q_1,\ldots,Q_t\}.$ This is again a finite set. Since $\overline{\mathcal X}(\overline{\mathbb{F}}_q)$ is infinite, we may choose a point $Q_{t+1}$ outside all the above finite sets. By its choice, for every $I$, we have $l_{\overline{\mathbb{F}}_q}(A_{I,t}+Q_{t+1})=0.$ Since $A_{I,t+1}=A_{I,t}+Q_{t+1},$ we obtain $l_{\overline{\mathbb{F}}_q}(A_{I,t+1})=0$ for every $I$. This completes the induction step.

After repeating the construction $h$ times, define $H:=Q_1+\cdots+Q_h.$ Then $H$ is an effective divisor satisfying $\deg(H)=h=d-e+\g-1$ and $\mathrm{supp}(H)\cap\bigl(\mathrm{supp}(D)\cup\mathrm{supp}(G)\bigr)=\varnothing.$ Moreover, for every $I\subseteq[n]$ with $|I|=n-e$, $\mathcal L_{\overline{\mathbb{F}}_q}(G+H-D_I)=\{0\}.$ Notice also that $\deg(G+H-D_I)=m+(d-e+\g-1)-(n-e)=\g-1.$

Each point $Q_j$ is defined over a finite extension of $\mathbb{F}_q$. Since there are only finitely many such points, there exists a common finite extension $K/\mathbb{F}_q$ over which all $Q_1,\ldots,Q_h$ are rational. Hence $H$ is an effective divisor on $\mathcal X_K$.

Finally, suppose for contradiction that, for some $I$, $0\neq f\in\mathcal L_K(G+H-D_I).$ Viewing $f$ as a function over $\overline{\mathbb{F}}_q$, we would also have $0\neq f\in\mathcal L_{\overline{\mathbb{F}}_q}(G+H-D_I),$ contradicting the construction above. Therefore, $\mathcal L_K(G+H-D_I)=\{0\}$ for every $I\subseteq[n]$ with $|I|=n-e$.
\end{proof}

In particular, if a function in $\mathcal L_K(G+H)$ vanishes at at least $n-e$ of the evaluation points, then it is zero. We use this consequence of Lemma~\ref{lem:auxiliary-divisor} in the following arguments.

\begin{lemma}\label{lem:interpolation-H}\label{lem:slack-bw-interpolant}
Fix the finite extension $K/\mathbb{F}_q$ and the divisor $H$ given by Lemma~\ref{lem:auxiliary-divisor}. Let $D_Z\ge 1$ be an integer satisfying $(d-2e)D_Z>\ell e.$ Then there exist a nonzero polynomial $A(\cdot,Z)\in\mathcal L_K(H)[Z]$ with $\deg_Z A\le D_Z-1$, and a polynomial $B(\cdot,Z)\in\mathcal L_K(G+H)[Z]$ with $\deg_Z B\le D_Z+\ell-1$, such that, for every $i\in[n]$, we have $A(P_i,Z)r_i(Z)=B(P_i,Z).$
\end{lemma}

\begin{proof}
Write
\[
A(\cdot,Z)=\sum_{j=0}^{D_Z-1}A_jZ^j,
\qquad
B(\cdot,Z)=\sum_{j=0}^{D_Z+\ell-1}B_jZ^j,
\] 
where $A_j\in\mathcal L_K(H)$ and $B_j\in\mathcal L_K(G+H).$

Since $\deg(H)=d-e+\g-1,$ the Riemann--Roch theorem gives $l_K(H)\ge \deg(H)+1-\g=d-e.$ Moreover, $\deg(G+H)=m+d-e+\g-1=n-e+\g-1,$ and hence $l_K(G+H)\ge \deg(G+H)+1-\g=n-e.$ Therefore, the number of unknown coefficients is at least $D_Z(d-e)+(D_Z+\ell)(n-e).$

For each $i\in[n]$, the identity $A(P_i,Z)r_i(Z)=B(P_i,Z)$ is an equality of polynomials of degree at most $D_Z+\ell-1$ in $Z$, and therefore gives at most $D_Z+\ell$ homogeneous linear equations over $K$. Thus, the total number of equations is at most $n(D_Z+\ell).$

By the assumption $(d-2e)D_Z>\ell e$, we have
\[
n(D_Z+\ell)=(n-e)(D_Z+\ell)+e(D_Z+\ell)<(n-e)(D_Z+\ell)+(d-e)D_Z.
\]
Hence the number of unknowns is strictly larger than the number of equations, so the interpolation system admits a nonzero solution $(A,B)$.

It remains to prove that $A\neq0$. Suppose, to the contrary, that $A=0$. Then the interpolation identities imply $B(P_i,Z)=0$ for every $i\in [n]$. Consequently, each coefficient $B_j$ vanishes at all the points $P_1,\ldots,P_n$. 

Fix any subset $I\subseteq[n]$ with $|I|=n-e$, and write $D_I:=\sum_{i\in I}P_i.$ Since $B_j\in\mathcal L_K(G+H)$ and $B_j(P_i)=0$ for every $i\in I$, we have $B_j\in\mathcal L_K(G+H-D_I).$ By the defining property of $H$ in Lemma~\ref{lem:auxiliary-divisor}, $\mathcal L_K(G+H-D_I)=\{0\}.$ Thus $B_j=0$ for every $j$, and hence $B=0$, contradicting the fact that $(A,B)$ is a nonzero solution. Therefore $A\neq0$.
\end{proof}

Fix the finite extension $K/\mathbb F_q$ and the effective divisor $H$ given by Lemma~\ref{lem:auxiliary-divisor}, and fix $A(\cdot,Z)$ and $B(\cdot,Z)$ as in Lemma~\ref{lem:interpolation-H}. For every $z\in S$, we regard $f_z$ naturally as an element of $\mathcal L_K(G)$ and write $A_z:=A(\cdot,z)\in\mathcal L_K(H)$ and $B_z:=B(\cdot,z)\in\mathcal L_K(G+H).$

\begin{lemma}\label{lem:vertical-specialization}
Using the above notation, for every $z\in S$, we have $B_z=A_zf_z.$
\end{lemma}

\begin{proof}
Fix $z\in S$. Since $d\bigl(\mathbf r(z),f_z(D)\bigr)\le e,$ there are at least $n-e$ indices $i\in[n]$ such that $r_i(z)=f_z(P_i).$ Choose a subset $I_z\subseteq[n]$ of size $n-e$ consisting only of such indices.

For every $i\in I_z$, the interpolation identities give $B_z(P_i)=A_z(P_i)r_i(z)=A_z(P_i)f_z(P_i).$ Thus, the function $B_z-A_zf_z$ vanishes at every $P_i$ with $i\in I_z$. Moreover, $B_z\in\mathcal L_K(G+H)$ and $A_zf_z\in\mathcal L_K(H+G)=\mathcal L_K(G+H).$ Therefore, $B_z-A_zf_z\in\mathcal L_K\left(G+H-\sum_{i\in I_z}P_i\right).$ By the defining property of $H$ in Lemma~\ref{lem:auxiliary-divisor}, the space is zero. Hence $B_z-A_zf_z=0,$ which proves that $B_z=A_zf_z.$
\end{proof}

For $z\in S$, define the error set $E_z:=\{i\in[n]: r_i(z)\ne f_z(P_i)\}.$ Note that $|E_z|\le e$. Define $T:=\{i\in[n]: A(P_i,Z)\equiv 0\}.$

\begin{lemma}\label{lem:incidence-bound}
Using the notation above, we have $\sum_{z\in S}|E_z\setminus T|\le n(D_Z-1).$
\end{lemma}

\begin{proof}
Now let $z\in S$ and $i\in E_z$. By Lemmas~\ref{lem:interpolation-H} and~\ref{lem:vertical-specialization}, we obtain $A(P_i,z)r_i(z)=B(P_i,z)=A(P_i,z)f_z(P_i).$ 
This implies that $A(P_i,z)(r_i(z)-f_z(P_i))=0.$

For each $i\in E_z$, we have $r_i(z)\neq f_z(P_i)$. This implies $A(P_i,z)=0.$ If $i\notin T$, then $A(P_i,Z)$ is a nonzero polynomial in $Z$ of degree at most $D_Z-1$. Thus, there are at most $D_Z-1$ distinct $z\in S$ with $A(P_i,z)=0$ for $i\in E_z\setminus T$. Thus, we have
\[
\sum_{z\in S}|E_z\setminus T| = \sum_{i\in[n]\setminus T} \bigl|\{z\in S:i\in E_z\}\bigr|\le(n-|T|)(D_Z-1)\le n(D_Z-1).
\]
\end{proof}

\begin{lemma}\label{lem:clean-collinearity}
If $|S|\ge\ell+1$ and $|S|>\frac{(\ell+1)n(D_Z-1)}{e+1}$, then there exist $v_0,\ldots,v_\ell\in \mathcal{L}(G)$ such that $f_z=\sum_{j=0}^{\ell}v_jz^j$ for every $z\in S$.
\end{lemma}

\begin{proof}
By Lemma~\ref{lem:incidence-bound}, averaging over all subsets of $S$ of size $\ell+1$ gives
\[
\begin{aligned}
\frac{1}{\binom{|S|}{\ell+1}}\sum_{\substack{W\subseteq S\\ |W|=\ell+1}}\sum_{z\in W}|E_z\setminus T|=\frac{\ell+1}{|S|}\sum_{z\in S}|E_z\setminus T|\le\frac{(\ell+1)n(D_Z-1)}{|S|}<e+1.
\end{aligned}
\]
Then there are $\ell+1$ distinct elements $z_0,\ldots,z_\ell\in S$ such that $\sum_{h=0}^{\ell}|E_{z_h}\setminus T|\le e.$ Define $v_0,\ldots,v_\ell$ by Lagrange interpolation:
\[
\sum_{j=0}^{\ell}v_jZ^j
:=\sum_{h=0}^{\ell} f_{z_h}
\prod_{\substack{0\le k\le\ell\\ k\ne h}}\frac{Z-z_k}{z_h-z_k}.
\]
Since $z_0,\ldots,z_\ell\in\mathbb F_q$ and $f_{z_0},\ldots,f_{z_\ell}\in\mathcal L(G)$, we have $v_0,\ldots,v_\ell\in\mathcal L(G)$.

We first show that
\begin{equation}\label{eq:ag-affine-interpolant}
B(\cdot,Z)=A(\cdot,Z)\sum_{j=0}^{\ell}v_jZ^j.
\end{equation}
Define $J:=\bigcup_{h=0}^{\ell}(E_{z_h}\setminus T)$, so that $|J|\le e$. If $i\in T$, then $A(P_i,Z)\equiv0$ and the interpolation identities give $B(P_i,Z)\equiv0$. If $i\notin T\cup J$, then $f_{z_h}(P_i)=r_i(z_h)$ for $0\le h\le\ell$. By the definition of $v_0,\ldots,v_\ell$, the polynomials $r_i(Z)$ and $\sum_{j=0}^{\ell}v_j(P_i)Z^j$ agree at the $\ell+1$ distinct parameters $z_0,\ldots,z_\ell$. They have degree at most $\ell$, so they are equal. Applying the interpolation identities, we obtain
\[
B(P_i,Z)-A(P_i,Z)\sum_{j=0}^{\ell}v_j(P_i)Z^j\equiv0
\qquad\text{for every }i\notin J.
\]
Every coefficient of $B(\cdot,Z)-A(\cdot,Z)\sum_{j=0}^{\ell}v_jZ^j$ belongs to $\mathcal L_K(G+H)$ and vanishes at at least $n-e$ evaluation points. By Lemma~\ref{lem:auxiliary-divisor}, every coefficient is zero. This proves~\eqref{eq:ag-affine-interpolant}.

For each $i\notin T$, combining~\eqref{eq:ag-affine-interpolant} with the interpolation identities gives
\[
A(P_i,Z)\left(r_i(Z)-\sum_{j=0}^{\ell}v_j(P_i)Z^j\right)\equiv0.
\]
Since $A(P_i,Z)\not\equiv0$, it follows that $r_i(Z)=\sum_{j=0}^{\ell}v_j(P_i)Z^j$ for every $i\notin T$.

Fix a nonzero coefficient $A_j$ of $A(\cdot,Z)$, which exists since $A\neq0$. Fix $z\in S$ and let $g_z:=\sum_{h=0}^{\ell}v_hz^h\in\mathcal L(G)$. At every coordinate where $f_z(P_i)=r_i(z)$, we have $A_j(P_i)(f_z(P_i)-g_z(P_i))=0$. This follows from $A_j(P_i)=0$ when $i\in T$, and from $r_i(z)=g_z(P_i)$ otherwise. Thus $A_j(f_z-g_z)\in\mathcal L_K(G+H)$ vanishes at at least $n-e$ evaluation points. Lemma~\ref{lem:auxiliary-divisor} gives $A_j(f_z-g_z)=0$. Since $A_j\neq0$, we obtain $f_z=g_z=\sum_{h=0}^{\ell}v_hz^h$ for every $z\in S$. This argument also applies when $A_z=0$.
\end{proof}

\begin{lemma}[\cite{goyal2026optimal} for $\ell=1$]\label{lem:collinearity-combinatorial}
Let $a\ge\ell+1$ and $\mathbf{c}_0,\ldots,\mathbf{c}_\ell\in\mathbb F_q^n$. Suppose there exist pairwise distinct elements $z_1,\ldots,z_a\in\mathbb F_q$ such that $d(\mathbf{r}(z_h),\sum_{j=0}^{\ell}\mathbf{c}_jz_h^j)\le e$ for every $h\in[a]$. Then $d([\mathbf{u}_0,\ldots,\mathbf{u}_\ell],[\mathbf{c}_0,\ldots,\mathbf{c}_\ell])\le\frac{a}{a-\ell}e.$
\end{lemma}

\begin{proof}
Let $w=d([\mathbf{u}_0,\ldots,\mathbf{u}_\ell],[\mathbf{c}_0,\ldots,\mathbf{c}_\ell])$. At each of these $w$ coordinates, $\sum_{j=0}^{\ell}(u_{j,i}-c_{j,i})Z^j$ is a nonzero polynomial of degree at most $\ell$, so it vanishes at at most $\ell$ of the $a$ parameters. Counting errors over these parameters gives $w(a-\ell)\le ae$, which proves the lemma.
\end{proof}

\begin{proof}[Proof of Theorem~\ref{thm:ag-corrlated-agreement}]
By the assumption on $|S|$, we have $|S|>\frac{(\ell+1)n(D_Z-1)}{e+1}$ and $|S|\ge a\ge\ell+1$. Thus Lemma~\ref{lem:clean-collinearity} gives $v_0,\ldots,v_\ell\in\mathcal L(G)$ such that $f_z=\sum_{j=0}^{\ell}v_jz^j$ for every $z\in S$. Consequently, $d(\mathbf r(z),\sum_{j=0}^{\ell}v_j(D)z^j)=d(\mathbf r(z),f_z(D))\le e$ for all $z\in S$. Applying Lemma~\ref{lem:collinearity-combinatorial} to any $a$ distinct elements of $S$ with $\mathbf c_j=v_j(D)$ for $0\le j\le\ell$, we obtain $d([\mathbf u_0,\ldots,\mathbf u_\ell],[v_0(D),\ldots,v_\ell(D)])\le\frac{a}{a-\ell}e.$ This proves the theorem.
\end{proof}

\begin{cor}\label{cor:constant}
Let $e$ be an integer with $0\le e<d/2$, let $0<\varepsilon<1$, and set
\begin{equation}\label{eq:ag-unique-parameters}
D_Z:=\left\lfloor\frac{\ell e}{d-2e}\right\rfloor+1,
\qquad
M_e:=\left\lfloor\frac{(\ell+1)n(D_Z-1)}{e+1}\right\rfloor.
\end{equation}
Writing $\delta=d/n$ and $\gamma=e/n$, we have $M_e<\ell(\ell+1)/(\delta-2\gamma)$. If $|S|\ge\max\{M_e+1,\,\ell+\lceil\ell/\varepsilon\rceil\},$ then there exist $v_0,\ldots,v_\ell\in\mathcal L(G)$ such that 
\[
d\bigl([\mathbf u_0,\ldots,\mathbf u_\ell],[v_0(D),\ldots,v_\ell(D)]\bigr)\le (1+\varepsilon)e.
\]

Consequently, $\mathcal C$ has $(\ell,e,\sigma_\varepsilon,\varepsilon)$-CA and $(\ell,e,\sigma_\varepsilon,\varepsilon)$-PG where $\sigma_\varepsilon:=\frac{\max\{M_e,\,\ell+\lceil\ell/\varepsilon\rceil-1\}}{q}=O\!\left( \frac{\max\{\ell(\ell+1)/(\delta-2\gamma),\,\ell/\varepsilon\}}{q}\right).$ For fixed $\ell,\delta,\gamma,\varepsilon$ with $\delta>2\gamma$, this bound is independent of $n$.
\end{cor}

\begin{proof}
The choice of $D_Z$ gives $(d-2e)D_Z>\ell e$. Apply Theorem~\ref{thm:ag-corrlated-agreement} with $a=\ell+\lceil\ell/\varepsilon\rceil$, for which $a/(a-\ell)\le1+\varepsilon$. The soundness bound follows from the integer threshold on $|S|$. Moreover,
\[
M_e\le\frac{(\ell+1)n}{e+1}\left\lfloor\frac{\ell e}{d-2e}\right\rfloor
\le\frac{\ell(\ell+1)ne}{(e+1)(d-2e)}<\frac{\ell(\ell+1)n}{d-2e}
=\frac{\ell(\ell+1)}{\delta-2\gamma}.
\]
The corresponding proximity gap statement follows from correlated agreement, as in Remark~\ref{rmk:mca-ca-pg}.
\end{proof}

\begin{rmk}\label{rmk:tower}
For the Garcia-Stichtenoth tower in Section~\ref{sec:GS-tower} over a fixed square field $\mathbb F_q$, take $n=n_j$ and $P_\infty=P_{\infty,j}$. Lemma~\ref{lem:GS-tower} gives $\g/n=1/(\sqrt q-1)+o(1)$ as $n\to\infty$. For fixed $2/(\sqrt q-1)<\rho<1$, taking $m=\lfloor\rho n\rfloor$ gives $m\ge2\g-1$ for all sufficiently large $n$ and rate tending to $\rho-1/(\sqrt q-1)$. For any fixed $0<\gamma<(1-\rho)/2$, take $e=\lfloor\gamma n\rfloor$. Then $(d-2e)/n$ is bounded away from zero, and for fixed $\ell$ and $\varepsilon\in(0,1)$, Corollary~\ref{cor:constant} gives a soundness bound independent of $n$. No additional genus term is required in this radius condition.
\end{rmk}

%\begin{rmk}\label{rmk:tower}
%For one-point codes, the above result is conditional on the existence of a curve having at least $n+1$ rational points. If one wishes to obtain an asymptotic family over a fixed field, one may take $q$ to be a square and use a tower of curves attaining the Drinfeld-Vl\u{a}du\c{t} bound~\cite{GarciaStichtenoth95}. More precisely, for each curve $\mathcal X$ in the tower, set $n=\#\mathcal X(\mathbb F_q)-1$ and reserve one rational point for $P_\infty$. Then, as $n\to\infty$ along the tower, $\g/n=1/(\sqrt q-1)+o(1)$. For fixed $\rho>2/(\sqrt q-1)$, taking $m=\lfloor\rho n\rfloor$ gives $m\ge2\g-1$ for all sufficiently large $n$ and rate tending to $\rho-1/(\sqrt q-1)$. For any fixed $0<\gamma<(1-\rho)/2$, take $e=\lfloor\gamma n\rfloor$. Then $(d-2e)/n$ is bounded away from zero, and for fixed $\ell$, Corollary~\ref{cor:constant} gives a soundness bound independent of $n$. No additional genus term is required in this radius condition.
%\end{rmk}

\begin{cor}\label{cor:without-slack}
Let $e$ be an integer with $0\le e<d/2$, and let $D_Z$ and $M_e$ be as in~\eqref{eq:ag-unique-parameters}. If $|S|\ge \max\{M_e+1,\ \ell(e+1)+1\},$ then there exist $v_0,\ldots,v_\ell\in\mathcal L(G)$ such that
\[
    d\bigl([\mathbf u_0,\ldots,\mathbf u_\ell],[v_0(D),\ldots,v_\ell(D)]\bigr)\le e.
\]
In particular, $\mathcal C$ has $(\ell,e,\sigma,0)$-CA and $(\ell,e,\sigma,0)$-PG where $\sigma:=\frac{\max\{M_e,\ell(e+1)\}}{q}.$

When $d=\delta n$ and $e=\gamma n$ with fixed $\ell$ and $\delta>2\gamma>0$ and sufficiently large $n$, this bound becomes $\min\{1,\ell(e+1)/q\}$.
\end{cor}

\begin{proof}
Apply Theorem~\ref{thm:ag-corrlated-agreement} with $a=\ell(e+1)+1$. Then
\[
d([\mathbf u_0,\ldots,\mathbf u_\ell],[v_0(D),\ldots,v_\ell(D)])
\le\frac{\ell(e+1)+1}{\ell e+1}e=e+\frac{\ell e}{\ell e+1}<e+1.
\]
Since the distance is an integer, it is at most $e$. The soundness bound follows from the integer threshold on $|S|$. For fixed $\ell$ and $\delta>2\gamma>0$, Corollary~\ref{cor:constant} gives $M_e<\ell(\ell+1)/(\delta-2\gamma)$, whereas $\ell(e+1)$ grows with $n$. This proves the final assertion.
\end{proof}

%Taking $\ell=1$ in the preceding results gives the corresponding statements for affine lines.

%\begin{rmk}
%The extension field $K$ is used only for the auxiliary divisor and interpolation. The parameter $z$ still range over $\mathbb F_q$, and the functions $v_0,v_1$ are obtained by $\mathbb F_q$-linear combinations of $f_{z_1},f_{z_2}\in\mathcal L(G)$. Thus the soundness denominator remains $q$, with no dependence on $[K:\mathbb F_q]$.
%\end{rmk}

\section{Mutual Correlated Agreement up to the Johnson bound}
\begin{comment}
Throughout this section, let $\mathcal X$ be a smooth, projective, geometrically irreducible curve over $\mathbb F_q$ of genus $\g$ and let $\mathbb F_q(\mathcal X)$ be its function field. We write $\mathcal X(\mathbb F_q)$ for the set of $\mathbb F_q$-rational points of $\mathcal X$. For a divisor $A$ on $\mathcal{X}$, its Riemann-Roch space is
\[
    \mathcal{L}(A)
    = \{0\}\cup\{f\in \mathbb F_q(\mathcal X)^{*} : (f)+A\geq 0\},
    \qquad
    l(A)=\dim_{\mathbb{F}_q}\mathcal{L}(A),
\]
where $(f)$ denotes the principal divisor of $f$.

Let $P_1,\ldots,P_n,P_{\infty}\in\mathcal{X}(\mathbb{F}_q)$ be pairwise distinct rational points and set $D=P_1+\cdots+P_n,\, G=mP_{\infty}$ where $m$ is a positive integer satisfying $2\g-1\leq m<n$. We consider the one-point algebraic-geometric code
\[
    \mathcal{C}=\cC_{\mathcal{L}}(D,G) =\{f(D):f\in\mathcal{L}(G)\}\subseteq\mathbb{F}_q^n,
\]
where $f(D):=(f(P_1),\ldots,f(P_n))$. Together with the Riemann-Roch theorem, this gives $\dim_{\mathbb{F}_q}\mathcal{C}=m+1-\g$ and $d(\mathcal{C})\geq n-m$, where $n-m$ is the designed minimum distance of $\mathcal{C}$. Write $\rho=m/n$. In this section, we establish mutual correlated agreement (MCA) for $\mathcal C$ up to the Johnson radius $1-\sqrt\rho$.

\end{comment}

We use the one-point algebraic geometric code $\mathcal C=\cC_{\mathcal L}(D,G)$ and the notation from Section~\ref{subsec:ag-codes}. Here we assume $m\geq1$ and $m\geq2\g-1$. Fix an integer $\ell\ge1$. We establish mutual correlated agreement for polynomial curves of degree at most $\ell$ up to the Johnson radius $1-\sqrt\rho$ where $\rho=m/n$.

We first state the main theorem and then prove it in several steps. Fix a real number $0<\eta<1-\sqrt{\rho}$ and set
\[
    s := \left\lceil\frac{2\sqrt{\rho}}{\eta}\right\rceil+2,\,
    N := \left\lceil(s+1)\sqrt{mn}\right\rceil+\g-1,\,
    D_Y := \left\lfloor\frac{N-\g}{m}\right\rfloor,\,
    D_Z := \ell sD_Y.
\]

%For this evaluation code, we write $(z,f,A)$ for the $e$ agreement witness $(z,f(D),A)$ from Definition~\ref{def:mutual-correlated-agreement} where $f\in\mathcal L(G)$. It is \emph{MCA-bad} precisely when no $v_0,v_1\in\mathcal L(G)$ satisfy $v_0(P_i)=u_{0,i}$ and $v_1(P_i)=u_{1,i}$ for every $i\in A$.

%For fixed $\mathbf u_0,\mathbf u_1\in\mathbb F_q^n$ and an integer $e\geq0$, an \emph{$e$ agreement witness} is a triple $(z,f,A)$ with $z\in\mathbb F_q$, $f\in\mathcal L(G)$ and $A\subseteq[n]$, such that $|A|\geq n-e$ and $f(P_i)=u_{0,i}+zu_{1,i}$ for every $i\in A$. We call this witness \emph{MCA-bad} if there are no $v_0,v_1\in\mathcal L(G)$ satisfying $v_0(P_i)=u_{0,i}$ and $v_1(P_i)=u_{1,i}$ for all $i\in A$. Thus, the two functions must agree with the received words on the given set $A$.

\begin{thm}
\label{thm:ag-johnson}
Let $e$ be an integer satisfying $0\leq e\leq(1-\sqrt\rho-\eta)n$, and fix $\mathbf u_0,\ldots,\mathbf u_\ell\in\mathbb F_q^n$. Let $S\subseteq\mathbb F_q$ be the set of parameters $z$ for which there exist $f\in\mathcal L(G)$ and $A\subseteq[n]$ such that $(z,f,A)$ is an \emph{MCA-bad} $e$ agreement witness. Then
\[
    |S|\leq B_e:=(2D_Y+1)D_Z+\frac{n-m}{n-m-e}\bigl(4mD_YD_Z+\ell D_Y\bigr)+\ell D_Y(e+1).
\]
Consequently, $\mathcal C$ has $(\ell,e,\sigma)$-MCA with $\sigma=\min\{1,B_e/q\}$.
\end{thm}

We first record how finitely many local coefficients determine evaluations at the points $P_1,\ldots,P_n$.

\begin{lemma}
\label{lem:ag-local-recovery}
Choose $P_0\in\mathcal{X}_{\overline{\mathbb{F}}_q}(\overline{\mathbb{F}}_q)\setminus\{P_{\infty},P_1,\ldots,P_n\}$ and a local parameter $t$ at $P_0$. For every integer $a\ge0$, the map
\[
    \mathcal L_{\overline{\mathbb F}_q}(aP_\infty)
    \longrightarrow(\overline{\mathbb F}_q)^{a+1},
    \qquad h\longmapsto\bigl([t^0]h,\ldots,[t^a]h\bigr),
\]
is $\overline{\mathbb F}_q$-linear and injective. Consequently, for each $i\in[n]$, there exist $\lambda_{i,0},\ldots,\lambda_{i,a}\in\overline{\mathbb F}_q$, independent of $h$, such that $h(P_i)=\sum_{j=0}^{a}\lambda_{i,j}[t^j]h$ for every $h\in\mathcal L_{\overline{\mathbb F}_q}(aP_\infty).$
\end{lemma}

\begin{proof}
The map is clearly $\overline{\mathbb{F}}_q$-linear. If its value at $h$ is zero, then $v_{P_0}(h)\ge a+1$, which is impossible for a nonzero $h\in \mathcal L_{\overline{\mathbb{F}}_q}(aP_\infty)$. Hence $h=0$, proving injectivity.

For each $i\in[n]$, the map $h\mapsto h(P_i)$ is a $\overline{\mathbb{F}}_q$-linear functional on $\mathcal L_{\overline{\mathbb{F}}_q}(aP_\infty)$. By injectivity, evaluation at $P_i$ defines a linear functional on the image of this map. Extending this functional to $(\overline{\mathbb{F}}_q)^{a+1}$, there exist $\lambda_{i,0},\ldots,\lambda_{i,a}\in \overline{\mathbb{F}}_q$ such that $h(P_i)=\sum_{j=0}^a \lambda_{i,j}[t^j]h$ for every $h\in\mathcal L_{\overline{\mathbb{F}}_q}(aP_\infty)$.
\end{proof}

\subsection{Guruswami-Sudan interpolation}
\label{subsec:ag-gs-interpolation}

Fix $e,\mathbf u_0,\ldots,\mathbf u_\ell$ as in Theorem~\ref{thm:ag-johnson} and let $S$ be the set defined there. The Guruswami-Sudan interpolation for AG codes constructs an interpolation polynomial for a fixed received word~\cite[Section~2.2]{BM25}. Here we keep the curve parameter $Z$ as an indeterminate and construct a polynomial $Q(Y,Z)$ for the entire curve. We bound its degree in $Z$ and require the interpolation conditions to hold as polynomial identities in $Z$. Thus, the $Z$-degree bound enters the counts of both unknown coefficients and linear equations over $\mathbb{F}_q$. This uses the same counting idea as the Reed-Solomon proximity gap argument in~\cite[Section~3.1]{BCHKS25}.

We first recall the notion of multiplicity.

\begin{defn}
\label{def:ag-multiplicity}
Let $P\in\mathcal{X}(\mathbb{F}_q)$ be a rational point, $r\in\mathbb{F}_q$ and let $Q(Y)\in\mathbb{F}_q(\mathcal{X})[Y]$ have coefficients regular at $P$. We say that $Q$ has a root of multiplicity at least $s\ge 1$ at $(P,r)$ if, for a local parameter $t$ at $P$, there exist $q_{a,b}\in\mathbb F_q$ such that
\[
Q(Y)=\sum_{\substack{a,b\ge 0\\a+b\ge s}}q_{a,b}t^a(Y-r)^b.
\] %The multiplicity is at least $s$ if the same expansion holds without the final nonvanishing condition.
\end{defn}

For each $i\in[n]$, fix a local parameter $t_i$ at $P_i$. We seek a polynomial $Q(Y,Z)\in\mathbb{F}_q(\mathcal{X})[Y,Z]$ of the form
\begin{equation}
\label{eq:ag-uniform-interpolant}
\begin{gathered}
    Q(Y,Z)=\sum_{j=0}^{D_Y}Q_j(Z)Y^j,\quad Q_j(Z)\in\mathcal{L}\bigl((N-jm)P_{\infty}\bigr)[Z], \quad \deg_Z Q_j\leq D_Z.
\end{gathered}
\end{equation}
Here $\mathcal{L}(A)[Z]$ denotes the polynomials in $Z$ whose coefficients belong to $\mathcal{L}(A)$. We impose the conditions
\begin{equation}
\label{eq:ag-uniform-multiplicity}
    [t_i^aV^b]Q\bigl(t_i,r_i(Z)+V,Z\bigr)=0
    \quad\text{in }\mathbb{F}_q[Z],
    \qquad i\in[n],\quad a,b\geq 0,\quad a+b<s.
\end{equation}
The notation $Q(t_i,Y,Z)$ means that each function coefficient of $Q$ is expanded at $P_i$ in the local parameter $t_i$. Here $[t_i^aV^b]$ denotes the coefficient of $t_i^aV^b$ in the formal expansion of $Q(t_i,r_i(Z)+V,Z)$, which is a polynomial in $Z$. Condition~\eqref{eq:ag-uniform-multiplicity} requires the coefficient of every power of $Z$ in this polynomial to vanish. After choosing bases for the Riemann-Roch spaces in~\eqref{eq:ag-uniform-interpolant}, these conditions form a homogeneous linear system over $\mathbb{F}_q$.

\begin{lemma}
\label{lem:ag-uniform-interpolation}
Using the above notations, there exists a nonzero polynomial $Q(Y,Z)$ satisfying~\eqref{eq:ag-uniform-interpolant} and~\eqref{eq:ag-uniform-multiplicity}. Moreover, every such nonzero polynomial satisfies $\deg_Y Q>0$.
\end{lemma}

\begin{proof}
The number of unknown coefficients over $\mathbb{F}_q$ in~\eqref{eq:ag-uniform-interpolant} is $(D_Z+1)\sum_{j=0}^{D_Y} l\bigl((N-jm)P_{\infty}\bigr).$ There are $n\binom{s+1}{2}$ triples $(i,a,b)$ with $i\in[n]$, $a,b\geq 0$ and $a+b<s$. For each such triple, the polynomial $[t_i^aV^b]Q\bigl(t_i,r_i(Z)+V,Z\bigr)$ has degree at most $D_Z+\ell D_Y$ in $Z$, since $\deg_Z Q_j\leq D_Z$ and $\deg_Z r_i\leq\ell$. Thus~\eqref{eq:ag-uniform-multiplicity} imposes at most $(D_Z+\ell D_Y+1)n\binom{s+1}{2}$ homogeneous linear equations over $\mathbb{F}_q$. By the Riemann-Roch bound,
\[
\begin{aligned}
    \sum_{j=0}^{D_Y}l\bigl((N-jm)P_{\infty}\bigr)\geq \sum_{j=0}^{D_Y}(N+1-\g-jm)\geq \frac{(N+1-\g)^2}{2m}\geq \frac{(s+1)^2n}{2}.
\end{aligned}
\]
The second inequality follows by evaluating the sum explicitly and using $1\leq N+1-\g-mD_Y\leq m$. The last inequality follows from $N+1-\g=\lceil(s+1)\sqrt{mn}\rceil$. Since $D_Z=\ell sD_Y$, we obtain
\[
\begin{aligned}
    &(D_Z+1)\sum_{j=0}^{D_Y}
    l\bigl((N-jm)P_{\infty}\bigr)
    -(D_Z+\ell D_Y+1)n\binom{s+1}{2}\\
    &\qquad\geq
    (D_Z+1)\frac{(s+1)^2n}{2}
    -(D_Z+\ell D_Y+1)\frac{s(s+1)n}{2}\\
    &\qquad=
    \frac{(s+1)n}{2}\bigl(D_Z+1-\ell sD_Y\bigr)
    =\frac{(s+1)n}{2}>0.
\end{aligned}
\]
Hence the number of unknowns exceeds the number of equations, so the homogeneous linear system has a nonzero solution.

Finally, the definition of $s$ and the inequality $\eta<1-\sqrt{\rho}$ give $s(1-\sqrt{\rho})>2\sqrt{\rho}$. Together with $\g\leq m<\sqrt{mn}$, this implies $ N<(s+1)\sqrt{mn}+\g<(s+2)\sqrt{mn}<sn.$ Suppose that a nonzero solution satisfies $\deg_Y Q=0$, so that $Q(Y,Z)=Q_0(Z)$. Taking $b=0$ in~\eqref{eq:ag-uniform-multiplicity} shows that every coefficient of $Q_0(Z)$ as a polynomial in $Z$ vanishes to order at least $s$ at each $P_i$. Each such coefficient therefore belongs to $\mathcal{L}(NP_{\infty}-sD)$. Since $\deg(NP_{\infty}-sD)=N-sn<0$, this space is $\{0\}$. Thus $Q_0(Z)=0$, a contradiction.
\end{proof}

We now apply the above Guruswami-Sudan interpolation to the set $S$. Fix a nonzero polynomial $Q(Y,Z)$ given by Lemma~\ref{lem:ag-uniform-interpolation}. For each $z\in S$, fix an \emph{MCA-bad} $e$ agreement witness $(z,f_z,A_z)$. Since $|A_z|\geq n-e$ and $f_z(P_i)=r_i(z)$ for every $i\in A_z$, we have $d(f_z(D),\mathbf r(z))\leq e$. We retain these same witnesses throughout the proof. After specializing $Z=z$, condition~\eqref{eq:ag-uniform-multiplicity} ensures that $Q(Y,z)$ has multiplicity at least $s$ at every $(P_i,r_i(z))$. Hence, whenever $f_z(P_i)=r_i(z)$, we have $v_{P_i}(f_z-r_i(z))\ge 1$. By the multiplicity condition at $(P_i,r_i(z))$, this implies $v_{P_i}(Q(f_z,z))\ge s.$ Since $d(f_z(D),\mathbf r(z))\le e$, this holds for at least $n-e$ points $P_i$. Our choice of parameters gives $s(n-e)-N>(s\eta-\sqrt{\rho})n-\g \geq\sqrt{mn}+2\eta n-\g>0$. Thus, the zero multiplicities would exceed the pole degree if $Q(f_z,z)$ were nonzero and therefore $Q(f_z,z)=0$ for every $z\in S$.

\subsection{Local Root Lifting and Resultant Analysis}
\label{subsec:ag-factorization}

Recall that $p=\mathrm{char} \mathbb{F}_q$. We factor $Q$ over $\overline{\mathbb{F}}_q(\mathcal{X})$ and take local expansions on the curve $\mathcal{X}_{\overline{\mathbb{F}}_q}$. The parameters $z\in S$ and the chosen functions $f_z\in\mathcal{L}(G)$ are still defined over $\mathbb{F}_q$.

Factor $Q$ in $\overline{\mathbb{F}}_q(\mathcal{X})[Y,Z]$ as $Q(Y,Z)=C(Z)\prod_{\nu=1}^{r}H_\nu(Y,Z)^{e_\nu},$ where $C(Z)\in \overline{\mathbb{F}}_q(\mathcal{X})[Z]\setminus\{0\}$, the polynomials $H_\nu$ are pairwise nonassociate irreducible factors of positive $Y$-degree and $e_\nu\geq 1$. In characteristic $p$, an irreducible factor of positive $Y$-degree may have zero derivative with respect to $Y$. For example, $Y^p-Z$ is irreducible but its $Y$-derivative is zero. Since the $Y$-derivative appears in the denominator in the Hensel lifting below, we first remove this inseparable part.

%Indeed, otherwise every $W$-exponent in $R_\nu$ would be divisible by $p$, so $H_\nu$ would be a polynomial in $Y^{p\mu_\nu}$ and $Z$, contradicting the maximality of $a_\nu$.

For each $\nu$, choose $a_\nu\geq 0$ as large as possible so that $H_\nu(Y,Z)$ is a polynomial in $Y^{p^{a_\nu}}$ and $Z$. Set $\mu_\nu=p^{a_\nu}$. With the change of variable $W=Y^{\mu_\nu}$, write $H_\nu(Y,Z)=R_\nu(Y^{\mu_\nu},Z)$ where $R_\nu(W,Z)\in \overline{\mathbb{F}}_q(\mathcal{X})[W,Z]$. Then $\partial_W R_\nu\neq 0$. Moreover, each $R_\nu(W,Z)$ remains irreducible in $\overline{\mathbb{F}}_q(\mathcal{X})[W,Z]$. Indeed,  a nontrivial factorization of $R_\nu$ would give one of $H_\nu$ after substituting $W=Y^{\mu_\nu}$. Since $R_\nu$ has positive $W$-degree, it is $W$-primitive and remains irreducible in $\overline{\mathbb{F}}_q(\mathcal{X})(Z)[W]$ by Gauss's lemma. Thus, $\gcd_{\overline{\mathbb{F}}_q(\mathcal{X})(Z)[W]}(R_\nu,\partial_W R_\nu)=1$. For notational simplicity, we now rename $W$ as $Y$ and write $R_\nu(Y,Z)$ for the polynomial just constructed. Thus $\partial_Y R_\nu(Y,Z)\neq 0$ and $\gcd_{\overline{\mathbb{F}}_q(\mathcal{X})(Z)[Y]}(R_\nu,\partial_Y R_\nu)=1$. The factorization of $Q$ becomes 
\[
Q(Y,Z)=C(Z)\prod_{\nu=1}^{r} R_\nu(Y^{\mu_\nu},Z)^{e_\nu}.
\]
Since $e_\nu\geq 1$ for every $\nu$, the factorization gives
\begin{equation}\label{eq:ag-factor-degrees}
    \sum_{\nu=1}^{r}\mu_\nu\deg_Y R_\nu \leq \deg_Y Q\leq D_Y,\quad \deg_Z C+\sum_{\nu=1}^{r}\deg_Z R_\nu \leq \deg_Z Q\leq D_Z.
\end{equation}

We next choose a common auxiliary point for the local expansions used in Hensel lifting. For a point $P_0\in\mathcal{X}(\overline{\F}_q)$ at which all function coefficients of $R_\nu$ are regular, we write $R_{\nu,P_0}(Y,Z):=R_\nu(P_0;Y,Z)\in\overline{\F}_q[Y,Z]$ for the polynomial obtained by evaluating each of these coefficients at $P_0$.

\begin{lemma}
\label{lem:ag-common-auxiliary-point}
There exists a point $P_0\in\mathcal{X}(\overline{\F}_q) \setminus\{P_{\infty},P_1,\ldots,P_n\}$ such that, for every $\nu\in\{1,\ldots,r\}$, all function coefficients of $R_\nu$ are regular at $P_0$ and we have $\deg_Y R_{\nu,P_0}=\deg_Y R_\nu$, $\deg_Y(\partial_Y R_{\nu,P_0})=\deg_Y(\partial_Y R_\nu)$ and $\mathrm{Res}_Y(R_{\nu,P_0},\partial_Y R_{\nu,P_0}) \in\overline{\F}_q[Z]\setminus\{0\}.$
\end{lemma}

\begin{proof}
Fix $\nu$ and write $R_\nu(Y,Z)=\sum_{a,b\geq 0}c_{a,b}Y^aZ^b$ where $c_{a,b}\in\overline{\F}_q(\mathcal{X})$ and only finitely many coefficients are nonzero. Each nonzero $c_{a,b}$ has only finitely many zeros and poles. Exclude all these points. At any remaining point $P_0$, every coefficient is regular and $c_{a,b}(P_0)\neq 0$ whenever $c_{a,b}\neq 0$. Thus $\deg_Y R_{\nu,P_0}=\deg_Y R_\nu$. Every nonzero coefficient of $\partial_Y R_\nu$ is of the form $a c_{a,b}$ with $p\nmid a$, so it also remains nonzero at $P_0$. Since evaluating coefficients commutes with differentiation in $Y$, we also obtain $\deg_Y(\partial_Y R_{\nu,P_0})=\deg_Y(\partial_Y R_\nu)$.

The irreducibility of $R_\nu$ and the condition $\partial_Y R_\nu\neq 0$ imply that $R_\nu$ and $\partial_Y R_\nu$ are coprime in $\overline{\F}_q(\mathcal{X})(Z)[Y]$. Hence their resultant is a nonzero polynomial in $Z$ and we write $\mathrm{Res}_Y(R_\nu,\partial_Y R_\nu) =A_0+A_1Z+\cdots+A_dZ^d$ where $A_0,\ldots,A_d\in\overline{\F}_q(\mathcal{X}).$ Choose one nonzero coefficient $A_j$ and further exclude its zeros and poles, so that $A_j(P_0)\neq 0$.

By the two degree equalities above, evaluating the Sylvester matrix at $P_0$ gives the Sylvester matrix of $R_{\nu,P_0}$ and $\partial_Y R_{\nu,P_0}$. All its entries are regular at $P_0$, and evaluation commutes with taking determinants since the determinant is a polynomial expression in the matrix entries. Therefore $\mathrm{Res}_Y(R_{\nu,P_0},\partial_Y R_{\nu,P_0})=A_0(P_0)+A_1(P_0)Z+\cdots+A_d(P_0)Z^d\neq 0.$

There are only finitely many factors $R_\nu$ and the exclusions above remove only finitely many points for each $\nu$. Since $\mathcal{X}(\overline{\F}_q)$ is infinite, we can choose $P_0$ outside all these excluded points and $\{P_{\infty},P_1,\ldots,P_n\}$. This point satisfies all the required conditions.
\end{proof}

Fix a point $P_0$ as in Lemma~\ref{lem:ag-common-auxiliary-point}. For each $\nu$, the leading coefficient $\mathrm{lc}_Y(R_{\nu,P_0})$ is a nonzero polynomial in $\overline{\F}_q[Z]$ of degree at most $\deg_Z R_\nu$. Moreover, Lemma~\ref{lem:ag-resultant-properties} gives
\[
    \deg_Z\mathrm{Res}_Y (R_{\nu,P_0},\partial_Y R_{\nu,P_0}) \leq (2\deg_Y R_\nu-1)\deg_Z R_\nu.
\]
Here we used $\deg_Y R_{\nu,P_0}=\deg_Y R_\nu$, $\deg_Y(\partial_Y R_{\nu,P_0})\leq\deg_Y R_\nu-1$ and the fact that evaluating coefficients and taking the derivative with respect to $Y$ do not increase the $Z$-degree. Moreover, the resultant is nonzero by the choice of $P_0$.

Remove from $S$ all parameters $z$ for which $C(z)=0$ or, for some $\nu$,
\[
    \mathrm{lc}_Y(R_{\nu,P_0})(z)=0
    \quad\text{or}\quad
    \mathrm{Res}_Y
    (R_{\nu,P_0},\partial_Y R_{\nu,P_0})(z)=0.
\]
Denote the remaining set by $S'$. As each of these nonzero polynomials has at most its degree
many roots in $\F_q$, using~\eqref{eq:ag-factor-degrees} and $\mu_\nu\geq 1$, we obtain
\[
\begin{aligned}
    |S\setminus S'|
    &\leq \deg_Z C
    +2\sum_{\nu=1}^{r}
    (\deg_Y R_\nu)(\deg_Z R_\nu)\\
    &\leq D_Z
    +2\left(\sum_{\nu=1}^{r}\deg_Y R_\nu\right)
       \left(\sum_{\nu=1}^{r}\deg_Z R_\nu\right)\\
    &\leq (2D_Y+1)D_Z.
\end{aligned}
\]
For every $z\in S'$, the factorization of $Q$, together with $Q(f_z,z)=0$ and $C(z)\neq 0$, implies that at least one index $\nu$ satisfies $R_\nu(f_z^{\mu_\nu},z)=0.$ Fix an ordering of the factors. For each $z\in S'$, choose the first index $\nu$ such that $R_\nu(f_z^{\mu_\nu},z)=0$ and assign $z$ to this $\nu$. Now fix a parameter $z$ assigned to $\nu$. Since $f_z\in\mathcal L(G)$ and $P_0\neq P_\infty$, the function $f_z$ is regular at $P_0$. Evaluating the above identity at $P_0$ gives $R_{\nu,P_0}\bigl(f_z(P_0)^{\mu_\nu},z\bigr)=0.$ Moreover, since $z\in S'$, we have $\mathrm{Res}_Y \bigl(R_{\nu,P_0},\partial_YR_{\nu,P_0}\bigr)(z)\neq0.$ Hence, the specialized polynomials $R_{\nu,P_0}(Y,z)$ and $(\partial_YR_{\nu,P_0})(Y,z)$ have no common root. Therefore, $(\partial_YR_{\nu,P_0})(f_z(P_0)^{\mu_\nu},z)\ne0.$ Thus, $f_z(P_0)^{\mu_\nu}$ is a simple root of $R_{\nu,P_0}(Y,z)$.

Choose a local parameter $t$ at the fixed point $P_0$ and fix $\nu\in\{1,\ldots,r\}$. Expanding the function coefficients of $R_\nu$ at $P_0$ gives $R_\nu(t,Y,Z)=\sum_{a\geq 0}R_{\nu,a}(Y,Z)t^a$ where $R_{\nu,a}\in\overline{\F}_q[Y,Z].$ For every $a\geq 0$, we have $\deg_Y R_{\nu,a}\leq\deg_Y R_\nu$ and $\deg_Z R_{\nu,a}\leq\deg_Z R_\nu$. The constant term in $t$ is precisely $R_{\nu,P_0}(Y,Z)=R_\nu(P_0;Y,Z)$. For a parameter $z$ assigned to $\nu$, let $y_z:=f_z(P_0)^{\mu_\nu}.$ As shown above, $R_{\nu,0}(y_z,z)=0$ and $\partial_YR_{\nu,0}(y_z,z)\neq0,$ so $y_z$ is a simple root of $R_{\nu,0}(Y,z)$.

Our goal is to lift this root from its value at $P_0$ to the full local expansion of $f_z^{\mu_\nu}$. Since $f_z^{\mu_\nu}$ is regular at $P_0$, write $f_z^{\mu_\nu}=y_z+c_1t+c_2t^2+\cdots .$ Moreover, $R_\nu(f_z^{\mu_\nu},z)=0$, so after expanding the function coefficients of $R_\nu$ at $P_0$, this identity becomes 
\[
R_\nu\!\left(t,\,y_z+c_1t+c_2t^2+\cdots,\,z\right)=0.
\]
Thus, starting from the initial value $y_z$, we would like to determine $c_1,c_2,\ldots$ successively by forcing the coefficients of $t,t^2,\ldots$ in this identity to vanish.

Rather than carrying out this lifting separately for each $z$, we perform it uniformly in two variables $Y$ and $Z$. For this purpose, we construct rational functions $\alpha_{\nu,j}(Y,Z)\in\overline{\mathbb F}_q(Y,Z)$ such that, for every parameter $z$ assigned to $\nu$, we have $[t^j]f_z^{\mu_\nu}=\alpha_{\nu,j}(y_z,z).$ Set $\alpha_{\nu,0}(Y,Z):=Y.$ For $j\ge1$, define recursively
\begin{equation}
\label{eq:ag-hensel-recursion}
    \alpha_{\nu,j}(Y,Z) = -\frac{[t^j]R_\nu\left(t,\, Y+\sum_{h=1}^{j-1}\alpha_{\nu,h}(Y,Z)t^h,\,  Z \right) }{\partial_Y R_{\nu,P_0}(Y,Z)}.
\end{equation}
Thus each $\alpha_{\nu,j}$ is a rational function in $\overline{\mathbb F}_q(Y,Z)$.

The idea of the recursion is simple. Suppose that $\alpha_{\nu,1},\ldots,\alpha_{\nu,j-1}$ have already been chosen and set $Y_{j-1}(t):=Y+\sum_{h=1}^{j-1}\alpha_{\nu,h}(Y,Z)t^h.$ By construction, after substituting $Y_{j-1}(t)$ for the $Y$-variable in $R_\nu(t,Y,Z)$, the coefficients of $t,\ldots,t^{j-1}$ have already been made zero. Thus we may write $R_\nu\bigl(t,Y_{j-1}(t),Z\bigr)=R_{\nu,P_0}(Y,Z)+A_j(Y,Z)t^j+O(t^{j+1}),$ where $A_j(Y,Z)=[t^j]R_\nu\bigl(t,Y_{j-1}(t),Z\bigr)$ is the current coefficient of $t^j$. We now add one more term and replace $Y_{j-1}(t)$ by $Y_j(t)=Y_{j-1}(t)+\alpha_{\nu,j}(Y,Z)t^j.$ Modulo $t^{j+1}$, this changes the $t^j$-coefficient by $\partial_YR_{\nu,P_0}(Y,Z)\,\alpha_{\nu,j}(Y,Z).$ Indeed, $\alpha_{\nu,j}t^j$ has order $j$, so only the linear term contributes to the coefficient of $t^j$ and higher-order terms are of order at least $j+1$. Therefore,
\[
R_\nu\bigl(t,Y_j(t),Z\bigr)=R_{\nu,P_0}(Y,Z)+\Bigl(A_j+\partial_YR_{\nu,P_0}\,\alpha_{\nu,j}\Bigr)t^j+O(t^{j+1}).
\]
Hence, choosing $\alpha_{\nu,j}=-\frac{A_j}{\partial_YR_{\nu,P_0}}$ makes the coefficient of $t^j$ vanish. This is exactly the recursion defined above.

The preceding discussion explains why the recursion removes the coefficients successively. The following lemma records this lifting property together with the denominator and degree bounds needed later.

\begin{lemma}
\label{lem:ag-hensel-lifting}
Using the above notations, the rational functions $\alpha_{\nu,j}$ defined in~\eqref{eq:ag-hensel-recursion} have the following properties.
\begin{enumerate}[label=(\roman*), leftmargin=*]
    \item For every $j\geq 1$, the product $(\partial_Y R_{\nu,P_0})^{2j-1}\alpha_{\nu,j}$ belongs to $\overline{\F}_q[Y,Z]$. Its degrees satisfy
    \[
    \begin{aligned}
    &\deg_Y\bigl( (\partial_Y R_{\nu,P_0})^{2j-1}\alpha_{\nu,j} \bigr)\leq (2j-1)(\deg_Y R_\nu-1)+1, \\
    &\deg_Z\bigl((\partial_Y R_{\nu,P_0})^{2j-1}\alpha_{\nu,j} \bigr) \leq (2j-1)\deg_Z R_\nu.
    \end{aligned}
    \]

    \item For every $z\in S'$ assigned to $\nu$ and every $j\ge 0$, $\alpha_{\nu,j}(y_z,z)$ is well defined and $[t^j]f_z^{\mu_\nu}=\alpha_{\nu,j}(y_z,z).$
    
\end{enumerate}
\end{lemma}

\begin{proof}
We follow the Hensel-lifting approach of~\cite[Appendix~A.4]{BCIKS23}. The preceding calculation shows that each step makes the coefficient of $t^j$ zero without changing the lower order coefficients. Hence, for every $M\geq 1$,
\[
    R_\nu\left( t,\,Y+\sum_{j=1}^{M}\alpha_{\nu,j}(Y,Z)t^j,\, Z \right) \equiv R_{\nu,P_0}(Y,Z) \pmod{t^{M+1}}
\]
in $\overline{\F}_q(Y,Z)[[t]]$.

For part~(i), we use induction on $j$. The case $j=1$ follows from $(\partial_Y R_{\nu,P_0})\alpha_{\nu,1}=-R_{\nu,1}$ and the degree bounds on $R_{\nu,1}$. Now let $j\geq 2$ and assume the assertion for all smaller indices. Write $V(t)=\sum_{h=1}^{j-1}\alpha_{\nu,h}(Y,Z)t^h$ and recall that $A_j=[t^j]\sum_{a=0}^{j}t^aR_{\nu,a}(Y+V(t),Z).$ Expand $R_{\nu,a}(Y+V(t),Z)$ by the binomial theorem. The polynomial multiplying $V(t)^u$ has $Y$-degree at most $\deg_Y R_\nu-u$ and $Z$-degree at most $\deg_Z R_\nu$. A contribution to $A_j$ selects one coefficient from each of the $u$ factors of $V(t)$, with indices summing to $j-a$.

By induction, each selected $\alpha_{\nu,h}$ can be written with denominator $(\partial_Y R_{\nu,P_0})^{2h-1}$ and numerator degrees at most $(2h-1)(\deg_Y R_\nu-1)+1$ in $Y$ and $(2h-1)\deg_Z R_\nu$ in $Z$. Thus their product has denominator $(\partial_Y R_{\nu,P_0})^{2(j-a)-u}$ and numerator degrees at most $\bigl(2(j-a)-u\bigr)(\deg_Y R_\nu-1)+u$ and $\bigl(2(j-a)-u\bigr)\deg_Z R_\nu$ in $Y$ and $Z$, respectively. For $u=0$, the product is $1$ and only $a=j$ contributes.

If $a=0$, then $u\ge2$, since $V(t)$ contains no $t^j$-term and a single factor of $V(t)$ therefore cannot contribute to the coefficient of $t^j$. If $a\ge1$, we also have $2a+u\ge2$. Hence in all cases, $2(j-a)-u\le 2j-2.$ Therefore all denominators are cleared by $(\partial_YR_{\nu,P_0})^{2j-2}$ and $(\partial_YR_{\nu,P_0})^{2j-2}A_j$ is a polynomial. For each term, this multiplication removes the denominator and leaves the factor $(\partial_Y R_{\nu,P_0})^{2a+u-2}$ in the numerator. Since $\partial_Y R_{\nu,P_0}$ has $Y$-degree at most $\deg_Y R_\nu-1$ and $Z$-degree at most $\deg_Z R_\nu$, the resulting $Y$- and $Z$-degrees are bounded, respectively, by
\[
\begin{aligned}
    (\deg_Y R_\nu-u)+u+(2j-2)(\deg_Y R_\nu-1)
    &=(2j-1)(\deg_Y R_\nu-1)+1,\\
    \deg_Z R_\nu+(2j-2)\deg_Z R_\nu
    &=(2j-1)\deg_Z R_\nu.
\end{aligned}
\]
These bounds also hold for the sum. Finally, $(\partial_Y R_{\nu,P_0})^{2j-1}\alpha_{\nu,j} =-(\partial_Y R_{\nu,P_0})^{2j-2}A_j,$ which proves part~(i).

For part~(ii), fix $z\in S'$ assigned to $\nu$. Since $(\partial_Y R_{\nu,P_0})(y_z,z)\neq 0$, part~(i) ensures that all $\alpha_{\nu,j}$ are defined at $(y_z,z)$. We prove $[t^j]f_z^{\mu_\nu}=\alpha_{\nu,j}(y_z,z)$ by induction on $j$. For $j=0$, both sides equal $y_z$.

Now let $j\geq 1$ and assume the equality for all smaller indices. Thus the coefficients of orders below $j$ in the local expansion of $f_z^{\mu_\nu}$ already agree with those given by the recursion. Using the calculation with $A_j$ above, the identity $R_\nu(t,f_z^{\mu_\nu},z)=0$ gives
\[
    0=[t^j]R_\nu(t,f_z^{\mu_\nu},z) =A_j(y_z,z) +(\partial_Y R_{\nu,P_0})(y_z,z)\,[t^j]f_z^{\mu_\nu}.
\]
Since $(\partial_Y R_{\nu,P_0})(y_z,z)\neq 0$, we obtain
\[
    [t^j]f_z^{\mu_\nu} =-\frac{A_j(y_z,z)}  {(\partial_Y R_{\nu,P_0})(y_z,z)} =\alpha_{\nu,j}(y_z,z),
\]
where the last equality is the defining recursion evaluated at $(y_z,z)$. This completes the induction.
\end{proof}

Keep $\nu$ fixed. We now use the lifted coefficients to express each equality $f_z(P_i)=r_i(z)$ as a polynomial equation in $(y_z,z)$. Since $f_z\in\mathcal L(G)$, we have $f_z^{\mu_\nu}\in\mathcal L_{\overline{\F}_q}(\mu_\nu mP_\infty)$. Applying Lemma~\ref{lem:ag-local-recovery} with $a=\mu_\nu m$ gives coefficients $\lambda_{i,0},\ldots,\lambda_{i,\mu_\nu m} \in\overline{\F}_q$, independent of $z$, such that
\begin{equation}
\label{eq:ag-recovered-values}
    f_z(P_i)^{\mu_\nu}
    =\sum_{j=0}^{\mu_\nu m}\lambda_{i,j}[t^j]f_z^{\mu_\nu}
    =\sum_{j=0}^{\mu_\nu m}\lambda_{i,j}
      \alpha_{\nu,j}(y_z,z)
\end{equation}
for every $z\in S'$ assigned to $\nu$, where the last equality follows from Lemma~\ref{lem:ag-hensel-lifting}(ii). Hence, only the coefficients $[t^j]f_z^{\mu_\nu}$ for $0\le j\le \mu_\nu m$ are needed to determine $f_z(P_i)^{\mu_\nu}$.

For each $i\in[n]$, define
\begin{equation}
\label{eq:ag-coordinate-polynomial}
    F_i(Y,Z)
    :=\bigl(\partial_Y R_{\nu,P_0}(Y,Z)\bigr)^{2\mu_\nu m}
    \left(
        \sum_{j=0}^{\mu_\nu m}\lambda_{i,j}\alpha_{\nu,j}(Y,Z)
        -r_i(Z)^{\mu_\nu}
    \right)\in \overline{\F}_q[Y,Z].
\end{equation}
The power of $\partial_Y R_{\nu,P_0}$ clears the denominators
in the sum. The following lemma gives the degree bounds and
relates $F_i$ to agreement at coordinate $i$.

\begin{lemma}
\label{lem:ag-coordinate-polynomials}
For every $i\in[n]$, the expression in~\eqref{eq:ag-coordinate-polynomial} defines a polynomial $F_i\in\overline{\F}_q[Y,Z]$ satisfying $\deg_Y F_i\leq 2\mu_\nu m\deg_Y R_\nu$ and $\deg_Z F_i\leq 2\mu_\nu m\deg_Z R_\nu+\ell\mu_\nu.$ Moreover, for every $z\in S'$ assigned to $\nu$, we have
\[
    F_i(y_z,z)=0
    \quad\Longleftrightarrow\quad
    f_z(P_i)=r_i(z).
\]
\end{lemma}

\begin{proof}
For $1\le j\le \mu_\nu m$, we have
\[
\begin{aligned}
\deg_Y\!\left(
(\partial_YR_{\nu,P_0})^{2\mu_\nu m}\alpha_{\nu,j}
\right)
\le
\deg_Y\!\left(
(\partial_YR_{\nu,P_0})^{2j-1}\alpha_{\nu,j}
\right) +
(2\mu_\nu m-2j+1)
\deg_Y(\partial_YR_{\nu,P_0}).
\end{aligned}
\]
By Lemma~\ref{lem:ag-hensel-lifting}(i) and $\deg_Y(\partial_YR_{\nu,P_0})\le \deg_YR_\nu-1,$ this is at most
\[
\begin{aligned}
&(2j-1)(\deg_YR_\nu-1)+1 +
(2\mu_\nu m-2j+1)(\deg_YR_\nu-1) \\&=
2\mu_\nu m(\deg_YR_\nu-1)+1 \le
2\mu_\nu m\deg_YR_\nu.
\end{aligned}
\]
Similarly, its $Z$-degree is at most
\[
\begin{aligned}
\deg_Z\!\left(
(\partial_YR_{\nu,P_0})^{2\mu_\nu m}\alpha_{\nu,j}
\right)
&\le
\deg_Z\!\left(
(\partial_YR_{\nu,P_0})^{2j-1}\alpha_{\nu,j}
\right) +
(2\mu_\nu m-2j+1)
\deg_Z(\partial_YR_{\nu,P_0}) \\
&\le
(2j-1)\deg_ZR_\nu
+
(2\mu_\nu m-2j+1)\deg_ZR_\nu \\
&=
2\mu_\nu m\deg_ZR_\nu.
\end{aligned}
\]
The same bounds hold for $j=0$, since $\alpha_{\nu,0}=Y$.

The remaining term, $(\partial_Y R_{\nu,P_0})^{2\mu_\nu m}r_i(Z)^{\mu_\nu}$, has $Y$-degree at most $2\mu_\nu m(\deg_Y R_\nu-1)$ and $Z$-degree at most $2\mu_\nu m\deg_Z R_\nu+\ell\mu_\nu$, because $\deg_Z r_i\leq\ell$. This proves that $F_i$ is a polynomial with the stated degree bounds.

For every $z\in S'$ assigned to $\nu$,
\eqref{eq:ag-recovered-values} gives
\[
    F_i(y_z,z)=\bigl((\partial_Y R_{\nu,P_0})(y_z,z)\bigr)^{2\mu_\nu m}\bigl(f_z(P_i)^{\mu_\nu}-r_i(z)^{\mu_\nu}\bigr).
\]
The first factor is nonzero. Since $\mu_\nu$ is a power of  the characteristic, the map $x\mapsto x^{\mu_\nu}$ is injective on $\overline{\F}_q$. Hence $F_i(y_z,z)=0$ if and only if $f_z(P_i)=r_i(z)$.
\end{proof}

Recall that each $z\in S'$ has been assigned to an index $\nu$ and satisfies $R_{\nu,P_0}(y_z,z)=0$. We now divide the parameters assigned to each $\nu$ into sets according to the irreducible factors of $R_{\nu,P_0}$. For fixed $\nu$, fix an irreducible factorization of the $Y$-primitive part of $R_{\nu,P_0}$ in $\overline{\mathbb F}_q[Y,Z]$ and let $\mathcal H_\nu$ be the set of all irreducible factors appearing in this factorization. Every $H\in\mathcal H_\nu$ has positive $Y$-degree. Since $\mathrm{Res}_Y(R_{\nu,P_0},\partial_Y R_{\nu,P_0})\neq 0$, none of these factors is repeated. The degree properties of the specialization at $P_0$ therefore give $|\mathcal H_\nu| \leq \deg_Y R_{\nu,P_0}=\deg_Y R_\nu$ and $\sum_{H\in\mathcal H_\nu}\deg_Z H \leq \deg_Z R_{\nu,P_0}\leq\deg_Z R_\nu.$

For every $z\in S'$ assigned to $\nu$, we have $\mathrm{lc}_Y(R_{\nu,P_0})(z)\neq 0$. Thus the $Y$-content of $R_{\nu,P_0}$ is nonzero at $z$ and $R_{\nu,P_0}(y_z,z)=0$ implies that $H(y_z,z)=0$ for some $H\in\mathcal H_\nu$. This factor is unique. Indeed, if two distinct factors vanished at $(y_z,z)$, this would imply $(\partial_Y R_{\nu,P_0})(y_z,z)=0$, contradicting the definition of $S'$. For each $H\in\mathcal H_\nu$, define
\[
    S_{\nu,H}:=\{z\in S':\ z\text{ is assigned to }\nu  \text{ and }H(y_z,z)=0\}.
\]
Since each parameter was assigned to exactly one $\nu$, the nonempty sets $S_{\nu,H}$ form a partition of $S'$. Moreover, by~\eqref{eq:ag-factor-degrees}, we have
\[
    \sum_{\nu=1}^{r}|\mathcal H_\nu|
    \leq\sum_{\nu=1}^{r}\deg_Y R_\nu
    \leq\sum_{\nu=1}^{r}\mu_\nu\deg_Y R_\nu
    \leq D_Y.
\]
For each $\nu\in\{1,\ldots,r\}$, let $F_i$ be the polynomials constructed in~\eqref{eq:ag-coordinate-polynomial} for this $\nu$. For every $H\in\mathcal H_\nu$, define $I_{\nu,H}:=\{i\in[n]: H\mid F_i\text{ in }\overline{\F}_q[Y,Z]\}$ and
\[
    B_{\nu,H} :=(\deg_Y H)\bigl(2\mu_\nu m\deg_Z R_\nu+\ell\mu_\nu\bigr)  +(\deg_Z H)\bigl(2\mu_\nu m\deg_Y R_\nu\bigr).
\]

\begin{thm}
\label{thm:ag-parameter-grouping}
With the notation above, for every $\nu\in\{1,\ldots,r\}$ and $H\in\mathcal H_\nu$, the following hold:
\begin{enumerate}[label=(\roman*), leftmargin=*]
    \item If $i\in I_{\nu,H}$, then $f_z(P_i)=r_i(z)$ for every $z\in S_{\nu,H}$.

    \item If $i\notin I_{\nu,H}$, then $\bigl|\{z\in S_{\nu,H}:f_z(P_i)=r_i(z)\}\bigr|\leq B_{\nu,H}.$
\end{enumerate}
Moreover, $\sum_{\nu=1}^{r}\sum_{H\in\mathcal H_\nu}B_{\nu,H} \leq 4mD_YD_Z+\ell D_Y.$
\end{thm}

\begin{proof}
Fix $\nu$ and $H\in\mathcal H_\nu$. If $i\in I_{\nu,H}$, then $H\mid F_i$. Since $H(y_z,z)=0$ for every $z\in S_{\nu,H}$, we also have $F_i(y_z,z)=0$. Lemma~\ref{lem:ag-coordinate-polynomials} gives $f_z(P_i)=r_i(z)$, proving~(i).

Now suppose $i\notin I_{\nu,H}$. Then $H\nmid F_i$ and $F_i\neq 0$. Since $H$ is irreducible and has positive $Y$-degree, Gauss's lemma implies that $H$ and $F_i$ are coprime in $\overline{\F}_q(Z)[Y]$. By Lemma~\ref{lem:ag-resultant-properties}, $\mathrm{Res}_Y(H,F_i)$ is a nonzero polynomial in $\overline{\F}_q[Z]$ and its degree satisfies
\[
\begin{aligned}
    \deg_Z\mathrm{Res}_Y(H,F_i)\leq (\deg_Y H)(\deg_Z F_i) +(\deg_Y F_i)(\deg_Z H)\leq B_{\nu,H},
\end{aligned}
\]
where the last inequality uses Lemma~\ref{lem:ag-coordinate-polynomials}. Whenever $z\in S_{\nu,H}$ satisfies $f_z(P_i)=r_i(z)$, we have $H(y_z,z)=F_i(y_z,z)=0$. Thus, $z$ is a root of $\mathrm{Res}_Y(H,F_i)$. This nonzero polynomial has at most $B_{\nu,H}$ distinct roots, which proves~(ii).

Finally, for each fixed $\nu$, summing $B_{\nu,H}$ over $H\in\mathcal H_\nu$ and using the degree bounds above gives
\[
\begin{aligned}
    \sum_{H\in\mathcal H_\nu}B_{\nu,H}
    &\leq (\deg_Y R_\nu)
          \bigl(2\mu_\nu m\deg_Z R_\nu+\ell\mu_\nu\bigr) +(\deg_Z R_\nu)
          \bigl(2\mu_\nu m\deg_Y R_\nu\bigr)\\
    &=4\mu_\nu m(\deg_Y R_\nu)(\deg_Z R_\nu)
      +\ell\mu_\nu\deg_Y R_\nu.
\end{aligned}
\]
Summing over $\nu$ and using~\eqref{eq:ag-factor-degrees}, we obtain
\[
\begin{aligned}
    \sum_{\nu=1}^{r}\sum_{H\in\mathcal H_\nu}B_{\nu,H}
    &\leq 4m\sum_{\nu=1}^{r}
          \mu_\nu(\deg_Y R_\nu)(\deg_Z R_\nu)
          +\ell\sum_{\nu=1}^{r}\mu_\nu\deg_Y R_\nu\\
    &\leq 4m\left(\sum_{\nu=1}^{r}\mu_\nu\deg_Y R_\nu\right)
            \left(\sum_{\nu=1}^{r}\deg_Z R_\nu\right)
          +\ell\sum_{\nu=1}^{r}\mu_\nu\deg_Y R_\nu\\
    &\leq 4mD_YD_Z+\ell D_Y.
\end{aligned}
\]
This completes the proof.
\end{proof}

\subsection{Proof of Theorem~\ref{thm:ag-johnson}}
\label{subsec:ag-main-proof}

To prove the main theorem, we first establish two counting lemmas.

\begin{lemma}
\label{lem:ag-large-group}
Let $\nu\in\{1,\ldots,r\}$ and $H\in\mathcal H_\nu$. If $e<n-m$ and $|S_{\nu,H}|> \frac{n-m}{n-m-e}\,B_{\nu,H},$ then $|I_{\nu,H}|>m$.
\end{lemma}

\begin{proof}
For every $z\in S_{\nu,H}$, we have $d(f_z(D),\mathbf r(z))\leq e$. Thus the number of pairs $(i,z)\in[n]\times S_{\nu,H}$ with $f_z(P_i)\neq r_i(z)$ is at most $e|S_{\nu,H}|$.

Suppose that $|I_{\nu,H}|\leq m$. We can then choose $n-m$ coordinates outside $I_{\nu,H}$. By Theorem~\ref{thm:ag-parameter-grouping}(ii), at each of these coordinates, the equality $f_z(P_i)=r_i(z)$ holds for at most $B_{\nu,H}$ parameters $z\in S_{\nu,H}$. Hence, for each such $i$, there are at least $|S_{\nu,H}|-B_{\nu,H}$ parameters $z\in S_{\nu,H}$ for which $f_z(P_i)\neq r_i(z).$ Counting the pairs $(i,z)$ with $i$ among these $n-m$ coordinates and $f_z(P_i)\neq r_i(z)$, we obtain $e|S_{\nu,H}| \geq (n-m)\bigl(|S_{\nu,H}|-B_{\nu,H}\bigr).$ Rearranging, it gives $(n-m-e)|S_{\nu,H}|\leq(n-m)B_{\nu,H},$ contrary to the assumed lower bound on $|S_{\nu,H}|$. Hence $|I_{\nu,H}|>m$.
\end{proof}

The next lemma uses a counting argument similar to that in~\cite[Lemma~2.4]{BCHKS25}.

\begin{lemma}
\label{lem:ag-common-to-correlated}
Fix $\mathbf u_0,\ldots,\mathbf u_\ell\in\mathbb F_q^n$, an integer $e\geq0$ and write $\mathbf r(z)=\sum_{j=0}^{\ell}\mathbf u_jz^j$. Let $S_0\subseteq\mathbb F_q$. For each $z\in S_0$, let $(z,f_z,A_z)$ be an \emph{MCA-bad} $e$ agreement witness. Suppose there is a set $I\subseteq[n]$ with $|I|>m$ such that $f_z(P_i)=r_i(z)$ for every $z\in S_0$ and $i\in I$. Then $|S_0|\leq\ell(e+1)$.
\end{lemma}

\begin{proof}
If $|S_0|\leq\ell$, the conclusion is immediate. Otherwise, choose distinct $z_0,\ldots,z_\ell\in S_0$ and define $v_0,\ldots,v_\ell$ by Lagrange interpolation:
\[
    \sum_{j=0}^{\ell}v_jZ^j
    :=\sum_{h=0}^{\ell}f_{z_h}
    \prod_{\substack{0\le k\le\ell\\ k\ne h}}\frac{Z-z_k}{z_h-z_k}.
\]
All these functions belong to $\mathcal L(G)$. For every $i\in I$, the equalities $f_{z_h}(P_i)=r_i(z_h)$ for $0\le h\le\ell$ give $v_j(P_i)=u_{j,i}$ for $0\le j\le\ell$. Thus, for every $z\in S_0$, the function $f_z-\sum_{j=0}^{\ell}v_jz^j\in\mathcal L(G)$ vanishes at all $P_i$ with $i\in I$. Since $|I|>m$, the bound on the number of zeros gives $f_z=\sum_{j=0}^{\ell}v_jz^j$ for every $z\in S_0$.

Let $w$ be the number of coordinates $i$ for which $(v_0(P_i),\ldots,v_\ell(P_i))\neq(u_{0,i},\ldots,u_{\ell,i})$. At each such coordinate, $\sum_{j=0}^{\ell}(v_j(P_i)-u_{j,i})Z^j$ is a nonzero polynomial of degree at most $\ell$. Hence $f_z(P_i)=r_i(z)$ holds for at most $\ell$ values of $z\in S_0$ at this coordinate.

For each $z\in S_0$, the witness $(z,f_z,A_z)$ is \emph{MCA-bad}, so $(v_0,\ldots,v_\ell)$ cannot agree with $(\mathbf u_0,\ldots,\mathbf u_\ell)$ on all of $A_z$. Therefore, at least one of these $w$ coordinates belongs to $A_z$ where $f_z(P_i)=r_i(z)$. Since each of the $w$ coordinates can satisfy this equality for at most $\ell$ parameters, we obtain $|S_0|\leq\ell w$.

On the other hand, each of these $w$ coordinates satisfies $f_z(P_i)\neq r_i(z)$ for at least $|S_0|-\ell$ parameters. Each parameter has at most $e$ such coordinates, because $|A_z|\geq n-e$. Counting the pairs $(i,z)\in[n]\times S_0$ with $f_z(P_i)\neq r_i(z)$ gives
\[
    |S_0|(|S_0|-\ell)\leq\ell w(|S_0|-\ell)\leq\ell e|S_0|.
\]
Dividing by $|S_0|$ proves $|S_0|\leq\ell(e+1)$.
\end{proof}

We are now ready to prove the main theorem.

\begin{proof}[Proof of Theorem~\ref{thm:ag-johnson}]
Suppose, for a contradiction, that $|S|$ exceeds the stated bound. Since $0<\rho<1$, the radius assumption gives $e\leq(1-\sqrt\rho-\eta)n<(1-\rho)n=n-m$. Recall that for every $z\in S$ we fixed an \emph{MCA-bad} witness $(z,f_z,A_z)$ and that $d(f_z(D),\mathbf r(z))\leq e$.

The preceding construction gives $S'\subseteq S$ with $|S\setminus S'|\leq(2D_Y+1)D_Z$. The nonempty sets $S_{\nu,H}$ partition $S'$ and $\sum_{\nu=1}^{r}|\mathcal H_\nu|\leq D_Y$. Together with Theorem~\ref{thm:ag-parameter-grouping} and the assumed lower bound on $|S|$, these estimates give
\[
\begin{aligned}
    |S'|&\geq |S|-(2D_Y+1)D_Z\\ &>\frac{n-m}{n-m-e}\bigl(4mD_YD_Z+\ell D_Y\bigr)
       +\ell D_Y(e+1)\\
    &\geq\sum_{\nu=1}^{r}\sum_{H\in\mathcal H_\nu}
       \left(\frac{n-m}{n-m-e}B_{\nu,H}+\ell(e+1)\right).
\end{aligned}
\]
Since $|S'|=\sum_{\nu=1}^{r}\sum_{H\in\mathcal H_\nu}|S_{\nu,H}|$, there is a set $S_{\nu,H}$ such that $|S_{\nu,H}|> \frac{n-m}{n-m-e}B_{\nu,H}+\ell(e+1).$ In particular, $|S_{\nu,H}|>\frac{n-m}{n-m-e}B_{\nu,H}$ and $|S_{\nu,H}|>\ell(e+1)$.

By Lemma~\ref{lem:ag-large-group}, we have $|I_{\nu,H}|>m$. Moreover, Theorem~\ref{thm:ag-parameter-grouping}(i) gives $f_z(P_i)=r_i(z)$ for every $z\in S_{\nu,H}$ and $i\in I_{\nu,H}$. For every $z$ in this set, retain the originally chosen \emph{MCA-bad} witness $(z,f_z,A_z)$. Applying Lemma~\ref{lem:ag-common-to-correlated} with $S_0=S_{\nu,H}$ and $I=I_{\nu,H}$ gives $|S_{\nu,H}|\leq\ell(e+1)$, a contradiction. This proves the theorem.
\end{proof}

\begin{rmk}
For fixed $\rho\in(0,1)$, our parameter choices give $s=O(\eta^{-1})$, $D_Y=O(\eta^{-1})$ and $D_Z=O(\ell\eta^{-2})$. Thus Theorem~\ref{thm:ag-johnson} gives $|S|=O(\ell n/\eta^3)$ and $(\ell,e,\sigma)$-MCA with $\sigma=\min\{1,B_e/q\}=O\!\left(\frac{\ell n}{\eta^3q}\right).$ By Remark~\ref{rmk:mca-ca-pg}, $\mathcal C$ also has $(\ell,e,\sigma,0)$-CA and $(\ell,e,\sigma,0)$-PG. Taking $\ell=1$ gives the corresponding affine line statements. %Allowing slack in correlated agreement does not, in general, reduce the soundness to $O(1/q)$ throughout the Johnson regime. Indeed, the Reed-Solomon counterexample recalled in~\cite[Section~1.1]{BCHKS25} has $\Omega(n)$ parameters satisfying $d(\mathbf r(z),\mathcal C)\le \frac n2,$ while $d([\mathbf u_0,\mathbf u_1],\mathcal C^2)\ge \frac{3n}{4}$ where the radius $1/2$ is strictly below the Johnson bound. Therefore, for every $\eta<1/2$, correlated agreement with slack $\varepsilon$ still has soundness error $\Omega(n/q)$, since $(1+\eta)\frac n2<\frac{3n}{4}.$ Since Reed-Solomon codes are the $\g=0$ case of one point AG codes, this obstruction applies to any correlated agreement statement covering all one point AG codes.

%It is well known that mutual correlated agreement implies correlated agreement, which in turn implies proximity gaps. Hence our theorem also gives correlated agreement and proximity gaps with the same soundness bound and without slack. However, for correlated agreement (and hence also for mutual correlated agreement), allowing a small slack does not generally reduce the number of exceptional parameters to $O(1)$ throughout the Johnson regime. Indeed, the Reed-Solomon counterexample recalled in~\cite[Section~1.1]{BCHKS25} has $\Omega(n)$ parameters satisfying $d(\mathbf{r}(z),\cC)\le n/2$, while $d([\mathbf u_0,\mathbf u_1],\mathcal C^2)\geq 3n/4$, where the radius $1/2$ is strictly below its Johnson bound. Thus, for correlated agreement and mutual correlated agreement, Thus, allowing multiplicative slack $0\leq\varepsilon<1/2$ still leaves $\Omega(n)$ exceptional parameters, since $(1+\varepsilon)n/2<3n/4$. Since Reed-Solomon codes are the $g=0$ case of our setting, this correlated agreement/mutual correlated agreement obstruction also applies to statements covering all one point AG codes.
\end{rmk}

\section{Counterexamples for Correlated Agreement}\label{sec:counterexample}
Previous work on Reed-Solomon codes has shown that, at certain radii above half the minimum distance but below the Johnson radius, there exist two words $\mathbf u_0,\mathbf u_1$ such that $\mathbf u_0+z\mathbf u_1$ is close to the code for many values of $z\in\mathbb F_q$, while $\mathbf u_0$ and $\mathbf u_1$ cannot simultaneously agree with two codewords on a sufficiently large common set of coordinates~\cite{DEEP20,BCHKS25}. Thus, even with fixed positive slack, the soundness error for correlated agreement can remain large.

In this section, we construct counterexamples for one-point AG codes on the Garcia-Stichtenoth tower. We take the natural point at infinity as $P_\infty$ and use all the other rational points for evaluation. Our first construction uses all nonzero trace values to show that correlated agreement can fail at radii strictly between half the actual minimum distance and the Johnson radius, even with fixed positive slack. Moreover, the relative gap above half the minimum distance can be made arbitrarily small. For each fixed positive bound on this gap and fixed slack less than one, the soundness error has a positive lower bound independent of $q$ and $n$. This constant may depend on the chosen gap and slack. A second construction keeps the field fixed while this gap tends to zero, and requires $\sigma=1$ for $(1,e,\sigma,\varepsilon)$-CA.

We use affine lines, that is, $\ell=1$. Taking $\mathbf u_2=\cdots=\mathbf u_\ell=0$ gives the same counterexamples for every $\ell\ge2$. We use the notation from Section~\ref{sec:preliminaries}. Let $\mathcal C=\mathcal C_{\mathcal L}(D,G)$, where $D=P_1+\cdots+P_n$, $G=mP_\infty$ and write $\rho=m/n$. Recall that $d(\mathcal C)$ denotes the minimum distance, while $n-m$ is the designed distance.

For $\mathbf u_0,\mathbf u_1\in\mathbb F_q^n$, write $\mathbf r(Z)=\mathbf u_0+Z\mathbf u_1$ and $S=\{z\in\mathbb F_q:d(\mathbf r(z),\mathcal C)\le e\}.$ To show that $(1,e,\sigma,\varepsilon)$-CA fails, it suffices to construct $\mathbf u_0,\mathbf u_1$ such that $\frac{|S|}{q}>\sigma$ and $d([\mathbf u_0,\mathbf u_1],\mathcal C^2)>(1+\varepsilon)e.$

For the rest of this section, we use the Garcia-Stichtenoth tower from Section~\ref{sec:GS-tower}, fix one level, omit the subscript $j$, and write $x=x_0$.

\subsection{A Counterexample Using a Nonzero Trace Value}

\begin{prop}
Let $q=p^{2h}$, where $p$ is prime and $h\ge2$. Suppose that $\frac{qb}{p^3}\le m<\frac{qb}{p^2}.$ Set $\mathbf u_0=x^{q/p}(D)$, $\mathbf u_1=x^{q/p^2}(D)$ and $e=n-\frac{qb}{p}.$ Then, for every $z\in(\mathbb F_q^*)^{p-1}$, we have $d(\mathbf r(z),\mathcal C)=e,$ while $d([\mathbf u_0,\mathbf u_1],\mathcal C^2)\ge n-\frac{qb}{p^2}\ge\left(1+\frac1p\right)e.$

Consequently, for every $0\le\varepsilon<1/p$, $(1,e,\sigma,\varepsilon)$-CA requires $\sigma\ge\frac{q-1}{(p-1)q}.$ Moreover, we have $\frac{d(\mathcal C)}2<e<(1-\sqrt\rho)n.$
\end{prop}

\begin{proof}
Let $T(Y)=\sum_{i=0}^{2h-1}Y^{p^i}.$ For every $\beta\in\mathbb F_q^*$, expanding the two highest-degree terms gives $T(\beta Y)=\beta^{q/p}Y^{q/p}+\beta^{q/p^2}Y^{q/p^2}+H_\beta(Y)$ where $\deg_Y H_\beta=\frac{q}{p^3}.$ Set $z=\beta^{q/p^2-q/p}$, fix one such $\beta$ for each $z$, and let $f_z=\frac{1-H_\beta(x)}{\beta^{q/p}}\in\mathcal L(G),$ where the inclusion follows from $m\ge qb/p^3$. For every $i\in[n]$, we have $r_i(z)-f_z(P_i)=\frac{T(\beta x(P_i))-1}{\beta^{q/p}}.$

Since $T(\beta a)=\operatorname{Tr}_{\mathbb F_q/\mathbb F_p}(\beta a)$, the equation $T(\beta a)=1$ has exactly $q/p$ solutions in $\mathbb F_q$, all of which are nonzero. By Lemma~\ref{lem:GS-tower}, each such solution corresponds to exactly $b$ evaluation points. Thus $\mathbf{r}(z)$ agrees with $f_z(D)$ on exactly $qb/p$ coordinates, giving $d(\mathbf r(z),\mathcal C)\le e.$ Moreover, $\gcd\left(\frac{q}{p^2}-\frac qp,q-1\right)=p-1.$ Hence, as $\beta$ ranges over $\mathbb F_q^*$, the values of $z$ range over $(\mathbb F_q^*)^{p-1}$, which has size $(q-1)/(p-1)$.

We next show that $d(\mathbf r(z),\mathcal C)\ge e$. Fix any $f\in\mathcal L(G)$. Since $m<qb/p^2<qb/p$, the function $x^{q/p}+zx^{q/p^2}-f$ has pole divisor $(qb/p)P_\infty$. It therefore has at most $qb/p$ zeros among the evaluation points, so $d(\mathbf{r}(z),f(D))\ge e$. Consequently, $d(\mathbf{r}(z),\mathcal C)=e$ for every $z\in(\mathbb F_q^*)^{p-1}$.

Similarly, for every $f\in\mathcal L(G)$, the function $x^{q/p^2}-f$ has pole divisor $(qb/p^2)P_\infty$. Hence $d([\mathbf u_0,\mathbf u_1],\mathcal C^2)\ge d(\mathbf u_1,\mathcal C)\ge n-\frac{qb}{p^2}.$ Since $e=n-qb/p$ and $n\le qb$, we obtain $n-\frac{qb}{p^2}=e+\left(\frac1p-\frac1{p^2}\right)qb\ge\left(1+\frac1p\right)e.$ Thus, for every $0\le\varepsilon<1/p$, the correlated agreement conclusion fails, while at least $(q-1)/(p-1)$ values of $z$ satisfy $d(\mathbf{r}(z),\mathcal C)\le e$. This proves the claimed soundness lower bound.

Finally, we verify the radius bounds. Since $m\ge qb/p^3\ge2b$, choose distinct $a_1,a_2\in\mathbb F_q^*$. The function $(x-a_1)(x-a_2)\in\mathcal L(G)$ gives a nonzero codeword with exactly $2b$ zero coordinates. Therefore, $d(\mathcal C)\le n-2b<2e,$ where the strict inequality follows from $n\ge(q-1)b$. On the other hand, $mn<\frac{qb}{p^2}n\le\frac{q^2b^2}{p^2}=(n-e)^2.$ Thus $e<(1-\sqrt{\rho})n,$ completing the proof.
\end{proof}

\begin{cor}
Fix a prime $p$, $\rho_0\in(p^{-3},p^{-2})$ and $0\le\varepsilon<1/p$. For every sufficiently large $q=p^{2h}$, consider the Garcia-Stichtenoth tower over $\mathbb F_q$. At each level $j$, set $m_j=\lfloor\rho_0n_j\rfloor$ and $e_j=n_j-\frac{qb_j}{p},$ and let $\mathcal C_j=\mathcal C_{\mathcal L}(D_j,m_jP_{\infty,j})$. For all sufficiently large $j$, these codes satisfy $2\g_j-1\le m_j<n_j$ and have rate tending to $\rho_0-\frac1{\sqrt q-1}>0.$ Moreover, $(1,e_j,\sigma,\varepsilon)$-CA requires soundness $\sigma\ge\frac{q-1}{(p-1)q}.$
\end{cor}

\subsection{Counterexamples Near Half the Minimum Distance}
\label{sec:trace-counterexamples}

We next construct counterexamples at radii that can be arbitrarily close to half the actual minimum distance. The key idea is to use the polynomial $T(Y)^{p^a-1}-1$, where $T$ is the trace polynomial to $\mathbb F_{p^a}$, so that all elements with nonzero trace become roots.

\begin{thm}\label{thm:GS-tower}
Let $p$ be a prime, $a\ge1$, $h\ge2$ be integers and set $q=p^{2ah}$. At a fixed level of the Garcia-Stichtenoth tower in Section~\ref{sec:GS-tower}, let $m=\left(1-\frac{2}{p^a}+\frac{1}{p^{3a}}\right)qb$ and $e=n-\left(1-\frac{1}{p^a}\right)qb,$ and consider the code $\mathcal C=\mathcal C_{\mathcal L}(D,mP_\infty)$. Set $\mathbf u_0=x^{(1-2/p^a+1/p^{2a})q}(D)$ and $\mathbf u_1=x^{(1-1/p^a)q}(D).$

For every $z\in-(\mathbb F_q^*)^{p^a-1}$, we have $d(\mathbf r(z),\mathcal C)=e,$ while $d([\mathbf u_0,\mathbf u_1],\mathcal C^2)\geq \left(2-\frac{1}{p^a}\right)e.$ Consequently, for every $0\le\varepsilon<1-1/p^a$, $(1,e,\sigma,\varepsilon)$-CA requires $\sigma\ge\frac{q-1}{(p^a-1)q}.$ Moreover, $0<\frac en-\frac{d(\mathcal C)}{2n}\le\frac{1}{2p^{3a}}$ and $e<(1-\sqrt\rho)n.$
\end{thm}

\begin{proof}
Let $T(Y)=\sum_{i=0}^{2h-1}Y^{p^{ai}}.$ For every $c\in\mathbb F_q$, we have $T(c)=\operatorname{Tr}_{\mathbb F_q/\mathbb F_{p^a}}(c).$ Thus, for every $\beta\in\mathbb F_q^*$, the polynomial $T(\beta Y)^{p^a-1}-1$ has exactly $(1-1/p^a)q$ roots in $\mathbb F_q$, all of which are nonzero.

We first expand its two highest-degree terms. Since $p^a-1=-1$ in $\mathbb F_q$, we obtain
\[
T(\beta Y)^{p^a-1}-1
=\beta^{(1-1/p^a)q}Y^{(1-1/p^a)q}-\beta^{(1-2/p^a+1/p^{2a})q}
Y^{(1-2/p^a+1/p^{2a})q}
+H_\beta(Y),
\]
where $\deg_Y H_\beta\le \frac{(p^a-2)q}{p^a}+\frac{q}{p^{3a}}=\frac mb.$ To see the degree bound, note that the highest degree term comes from choosing $Y^{q/p^a}$ in all $p^a-1$ factors. The second highest comes from choosing $Y^{q/p^a}$ in $p^a-2$ factors and $Y^{q/p^{2a}}$ in the remaining factor. Every other term either uses a monomial of degree at most $q/p^{3a}$, or uses $Y^{q/p^a}$ in at most $p^a-3$ factors. In the latter case, its degree is at most $\frac{(p^a-3)q}{p^a}+\frac{2q}{p^{2a}}<\frac{(p^a-2)q}{p^a}+\frac{q}{p^{3a}}.$ %When $p^a=2$, only the first case occurs.

Now set $z=-\beta^{q/p^a-q/p^{2a}}$ and $f_z=\frac{H_\beta(x)}{\beta^{(1-2/p^a+1/p^{2a})q}}.$ Since $\deg_Y H_\beta\le m/b$ and $(x)_\infty=bP_\infty$, we have $f_z\in\mathcal L(G)$. The expansion above gives, for every $i\in[n]$,
\[
r_i(z)-f_z(P_i)
=-\frac{T(\beta x(P_i))^{p^a-1}-1}
{\beta^{(1-2/p^a+1/p^{2a})q}}.
\]
The condition $T(\beta c)^{p^a-1}=1$, equivalently $T(\beta c)\ne0$, holds for exactly $(1-1/p^a)q$ values of $c\in\mathbb F_q$, all of which are nonzero. By Lemma~\ref{lem:GS-tower}, each such value corresponds to exactly $b$ evaluation points. Hence $\mathbf r(z)$ agrees with $f_z(D)$ on exactly $(1-1/p^a)qb$ coordinates, giving $d(\mathbf r(z),\mathcal C)\le e.$ Moreover, $\gcd\left(\frac q{p^a}-\frac q{p^{2a}},q-1\right)=p^a-1.$ Thus, as $\beta$ ranges over $\mathbb F_q^*$, the values of $z$ range over $-(\mathbb F_q^*)^{p^a-1}$, which has size $(q-1)/(p^a-1)$.

We next show that $d(\mathbf r(z),\mathcal C)\ge e$. Fix any $f\in\mathcal L(G)$. Since $m<\left(1-\frac2{p^a}+\frac1{p^{2a}}\right)qb<\left(1-\frac1{p^a}\right)qb,$ the function $x^{(1-2/p^a+1/p^{2a})q}+zx^{(1-1/p^a)q}-f$ has pole divisor $\left(1-\frac1{p^a}\right)qbP_\infty.$ It can therefore vanish at no more than $(1-1/p^a)qb$ evaluation points. Hence $d(\mathbf r(z),f(D))\ge e$, and together with the upper bound above, we obtain $d(\mathbf r(z),\mathcal C)=e$ for every $z\in-(\mathbb F_q^*)^{p^a-1}$.

Next, since $m<\left(1-\frac2{p^a}+\frac1{p^{2a}}\right)qb,$ no function in $\mathcal L(G)$ can cancel the highest-order pole of $x^{(1-2/p^a+1/p^{2a})q}$. Therefore,
\[
d([\mathbf u_0,\mathbf u_1],\mathcal C^2)\ge d(\mathbf u_0,\mathcal C)\ge n-\left(1-\frac2{p^a}+\frac1{p^{2a}}\right)qb.
\]
Since $n\le qb$, we have $e\le qb/p^a$, and thus
\[
n-\left(1-\frac2{p^a}+\frac1{p^{2a}}\right)qb=e+\left(\frac1{p^a}-\frac1{p^{2a}}\right)qb\ge\left(2-\frac1{p^a}\right)e.
\]
Consequently, for every $0\le\varepsilon<1-1/p^a$, the correlated agreement conclusion fails, while at least $(q-1)/(p^a-1)$ values of $z$ satisfy $d(\mathbf r(z),\mathcal C)\le e$. This proves the claimed soundness lower bound.

Finally, we verify the minimum distance and the radius bounds. Since $m/b$ is an integer smaller than $q-1$, choose a set $A\subseteq\mathbb F_q^*$ of size $m/b$. The function $f_A(x)=\prod_{c\in A}(x-c)\in\mathcal L(G)$ vanishes at exactly $m$ evaluation points. Hence $d(\mathcal C)\le n-m$, and the designed distance bound gives $d(\mathcal C)=n-m.$ Substituting the definitions of $m$ and $e$, we obtain $e-\frac{d(\mathcal C)}2=\frac{n-qb+qb/p^{3a}}2.$ Since $(q-1)b\le n\le qb$ and $q>p^{3a}$, this quantity is strictly positive. Moreover, $0<\frac en-\frac{d(\mathcal C)}{2n}\le\frac1{2p^{3a}}.$ To verify the Johnson bound, observe that $1-\frac2{p^a}+\frac1{p^{3a}}<\left(1-\frac1{p^a}\right)^2.$ Using $n\le qb$, we obtain
\[
mn\le\left(1-\frac2{p^a}+\frac1{p^{3a}}\right)q^2b^2<\left(1-\frac1{p^a}\right)^2q^2b^2
=(n-e)^2.
\]
Thus $e<(1-\sqrt\rho)n$, completing the proof.
\end{proof}

\begin{cor}
Fix a prime $p$ and an integer $a\ge1$. For every sufficiently large $q=p^{2ah}$, consider the codes $\mathcal C_j$ obtained from Theorem~\ref{thm:GS-tower} at each level $j$ of the Garcia-Stichtenoth tower. For all sufficiently large $j$, these codes satisfy $2\g_j-1\le m_j<n_j$, have $\g_j=\Theta(n_j)$, and have positive rate with $\liminf_{j\to\infty}\frac{\dim\mathcal C_j}{n_j}\ge1-\frac{2}{p^a}+\frac{1}{p^{3a}}-\frac{1}{\sqrt q-1}>0.$

Moreover, for every $0\le\varepsilon<1-1/p^a$, $(1,e_j,\sigma,\varepsilon)$-CA requires $\sigma\ge\frac{q-1}{(p^a-1)q}.$ The radius $e_j$ satisfies $0<\frac{e_j}{n_j}-\frac{d(\mathcal C_j)}{2n_j}\le\frac{1}{2p^{3a}},$ while its relative distance below the Johnson radius is at least $1-\frac1{p^a}-\sqrt{1-\frac2{p^a}+\frac1{p^{3a}}}>0.$
\end{cor}

Since $(q-1)b\le n\le qb$, we have $\rho=\frac{m}{n}=1-\frac{2}{p^a}+\frac{1}{p^{3a}}+O(1/q)$ and $\frac{e}{n}=\frac{1}{p^a}+O(1/q).$ The following table gives several choices of $p^a$. The values of $\rho$ and $e/n$ are their leading terms as $q$ grows, and the gap bound refers to $e/n-d(\mathcal C)/(2n)$.

\begin{table}[t]
\centering
\begingroup
\renewcommand{\arraystretch}{1.2}
\setlength{\tabcolsep}{10pt}
\begin{tabular}{ccccc}
\toprule
$p^a$ & $\rho$ & $e/n$ & Gap bound & Slack range \\
\midrule
$2$ & $\tfrac18$ & $\tfrac12$ & $\tfrac1{16}$
    & $\varepsilon<\tfrac12$ \\
$3$ & $\tfrac{10}{27}$ & $\tfrac13$ & $\tfrac1{54}$
    & $\varepsilon<\tfrac23$ \\
$4$ & $\tfrac{33}{64}$ & $\tfrac14$ & $\tfrac1{128}$
    & $\varepsilon<\tfrac34$ \\
$8$ & $\tfrac{385}{512}$ & $\tfrac18$ & $\tfrac1{1024}$
    & $\varepsilon<\tfrac78$ \\
\bottomrule
\end{tabular}
\endgroup
\caption{Selected parameters for the counterexamples in Theorem~\ref{thm:GS-tower}. The entries for $\rho$ and $e/n$ are their leading terms as $q$ grows, and the gap bound refers to $e/n-d(\mathcal C)/(2n)$.}
\label{tab:gs-ca-counterexamples}
\end{table}

By choosing $a$ sufficiently large, we can make the relative gap above half the actual minimum distance arbitrarily small. For each fixed $a$, the soundness error is bounded below by a positive constant independent of $q$ and $n$, but this lower bound decreases as $a$ grows.

\begin{cor}\label{cor:counterexample}
Fix a prime $p$, $0\le\varepsilon<1$, and $\zeta>0$. For sufficiently large integers $a$ and $h$, let $q=p^{2ah}$. The Garcia-Stichtenoth tower over $\mathbb F_q$ gives a family of one-point AG codes $\mathcal C_j$ with positive rate, unbounded length, and $\g_j=\Theta(n_j)$.

For all sufficiently large $j$, these codes satisfy $2\g_j-1\le m_j<n_j$, and their counterexample radii satisfy $0<\frac{e_j}{n_j}-\frac{d(\mathcal C_j)}{2n_j}<\zeta$ and $\frac{e_j}{n_j}<1-\sqrt{\rho_j}.$ Nevertheless, correlated agreement at radius $e_j$ with slack $\varepsilon$ requires $\sigma\ge\frac{q-1}{(p^a-1)q}\ge\frac{1}{2(p^a-1)}.$
\end{cor}

The counterexamples in this subsection use trace polynomials and can have relatively low code rates. The next subsection instead adapts the subspace annihilator polynomial construction of~\cite{LXYZ26}. Although this construction has higher code rates, it gives stronger soundness lower bounds: $\sigma=1$ for $(1,e,\sigma,\varepsilon)$-CA and $\sigma\ge 1-1/q$ for $(1,e,\sigma,\varepsilon)$-PG. Moreover, over a fixed field, the relative gap above half the actual minimum distance tends to zero as the code length grows.

\subsection{Counterexamples over a Fixed Field}\label{sec:fixed-field-counterexamples}

We next give counterexamples for which the relative gap above half the actual minimum distance tends to zero over a fixed field. Our construction is inspired by the subspace-polynomial counterexample for Gabidulin codes~\cite[Section~3.3]{LXYZ26}. Here we use the polynomial $(Y^q-Y)/A_U(Y)$, where $A_U(Y)=\prod_{u\in U}(Y-u)$ and $U$ is a two-dimensional subspace over a subfield. Its roots are exactly the elements outside $U$. We first recall that the second highest coefficient of $A_U$ can take every value in $\mathbb F_q$.

\begin{lemma}[Lemma~3.10~\cite{LXYZ26}]\label{lem:two-dimensional-coefficients}
Let $Q$ be a prime power and set $q=Q^6$. For each two-dimensional $\mathbb F_Q$-subspace $U\subseteq\mathbb F_q$, write
\[
A_U(Y)=\prod_{u\in U}(Y-u)=Y^{Q^2}+\alpha(U)Y^Q+\beta(U)Y.
\]
Then $\{\alpha(U):U\subseteq\mathbb F_q,\ \dim_{\mathbb F_Q}U=2\}=\mathbb F_q$.
\end{lemma}

\begin{thm}\label{thm:GS-fixed-field}
Let $Q$ be a prime power and set $q=Q^6$. At a fixed level of the Garcia-Stichtenoth tower in Section~\ref{sec:GS-tower}, let $m=(q-2Q^2+1)b$ and $e=n-(q-Q^2)b.$ The code $\mathcal C=\mathcal C_{\mathcal L}(D,mP_\infty)$ has minimum distance $d(\mathcal C)=n-m$. Moreover, there exist words $\mathbf u_0,\mathbf u_1\in\mathbb F_q^n$ such that $d(\mathbf u_0+z\mathbf u_1,\mathcal C)=e$ for every $z\in\mathbb F_q,$ but
\[
d([\mathbf u_0,\mathbf u_1],\mathcal C^2)
\ge d(\mathbf u_1,\mathcal C)
\ge e+(Q^2-Q)b.
\]
Consequently, for every $0\le\varepsilon<(Q^2-Q)b/e$, $(1,e,\sigma,\varepsilon)$-CA requires $\sigma=1$.
\end{thm}

\begin{proof}
Set
\[
\mathbf u_0=x^{q-Q^2}(D),\qquad
\mathbf u_1=x^{q-2Q^2+Q}(D).
\]
By Lemma~\ref{lem:two-dimensional-coefficients}, for every $z\in\mathbb F_q$ there is a two-dimensional $\mathbb F_Q$-subspace $U\subseteq\mathbb F_q$ whose annihilator polynomial has the form $A_U(Y)=Y^{Q^2}-zY^Q+\beta Y.$ Consider
\[
F_U(Y):=\frac{Y^q-Y}{A_U(Y)}
=\prod_{c\in\mathbb F_q\setminus U}(Y-c).
\]
Polynomial division gives $F_U(Y)=Y^{q-Q^2}+zY^{q-2Q^2+Q}+H_U(Y)$ with $\deg_Y H_U\le q-2Q^2+1=\frac mb.$ Indeed, subtracting $A_U(Y)(Y^{q-Q^2}+zY^{q-2Q^2+Q})$ from $Y^q-Y$ leaves a polynomial of degree at most $q-Q^2+1$. Dividing by $A_U$ gives the degree bound.

For such a $z$, define $f_z=-H_U(x)$. Since $(x)_\infty=bP_\infty$ and $\deg_Y H_U\le m/b$, we have $f_z\in\mathcal L(G)$ and $\mathbf u_0+z\mathbf u_1-f_z(D)=F_U(x)(D).$ The roots of $F_U$ are exactly the $q-Q^2$ elements of $\mathbb F_q\setminus U$, all of which are nonzero. By Lemma~\ref{lem:GS-tower}, each root corresponds to exactly $b$ evaluation points. Thus
\[
d(\mathbf u_0+z\mathbf u_1,f_z(D))
=n-(q-Q^2)b=e,
\]
and hence $d(\mathbf r(z),\mathcal C)\le e$.

We next show that $d(\mathbf r(z),\mathcal C)\ge e$. Fix any $f\in\mathcal L(G)$. Since $m<(q-2Q^2+Q)b<(q-Q^2)b$, the function $x^{q-Q^2}+zx^{q-2Q^2+Q}-f$ has pole divisor $(q-Q^2)bP_\infty$. It therefore vanishes at no more than $(q-Q^2)b$ evaluation points, giving $d(\mathbf r(z),f(D))\ge e$. Together with the upper bound, this proves $d(\mathbf r(z),\mathcal C)=e$ for every $z\in\mathbb F_q$.

To bound $d([\mathbf{u}_0,\mathbf{u}_1],\cC^2)$, it suffices to consider $\mathbf u_1$. For every $f\in\mathcal L(G)$, the function $x^{q-2Q^2+Q}-f$ has pole divisor $(q-2Q^2+Q)bP_\infty$. Therefore,
\[
d([\mathbf u_0,\mathbf u_1],\mathcal C^2)
\ge d(\mathbf u_1,\mathcal C)
\ge n-(q-2Q^2+Q)b
=e+(Q^2-Q)b.
\]
For $0\le\varepsilon<(Q^2-Q)b/e$, this is greater than $(1+\varepsilon)e$, while all $q$ parameters satisfy $d(\mathbf r(z),\mathcal C)=e$. Thus $(1,e,\sigma,\varepsilon)$-CA requires $\sigma=1$.

Finally, since $m/b$ is an integer smaller than $q-1$, choose a set $A\subseteq\mathbb F_q^*$ of size $m/b$. The function $f_A(x)=\prod_{c\in A}(x-c)\in\mathcal L(G)$ vanishes at exactly $m$ evaluation points. Hence $d(\mathcal C)\le n-m$, and the designed distance bound gives $d(\mathcal C)=n-m$.
\end{proof}

The same construction also gives a lower bound on the soundness error for proximity gaps. For the same range of $\varepsilon$, consider the line $\widetilde{\mathbf r}(z)=\mathbf u_1+z\mathbf u_0$. Since $\mathcal C$ is $\mathbb F_q$-linear, for every $z\ne0$ we have
\[
d(\widetilde{\mathbf r}(z),\mathcal C)
=d(\mathbf u_0+z^{-1}\mathbf u_1,\mathcal C)=e,
\]
whereas $d(\widetilde{\mathbf r}(0),\mathcal C) =d(\mathbf u_1,\mathcal C)>(1+\varepsilon)e$. Thus exactly $q-1$ parameters give points within distance $e$, so $(1,e,\sigma,\varepsilon)$-PG requires $\sigma\ge1-1/q$.

To compare $e$ with half the minimum distance and the Johnson radius, write $t=n-(q-1)b$, which is positive by Lemma~\ref{lem:GS-tower}. Substituting the definitions of $m$ and $e$ gives $e=(Q^2-1)b+t$ and $d(\mathcal C)=2(Q^2-1)b+t.$ Consequently, $0<\frac en-\frac{d(\mathcal C)}{2n}=\frac{t}{2n}.$ Moreover,
\[
(n-e)^2-mn
=(Q^2-1)^2b^2-(q-2Q^2+1)bt.
\]
Thus $e<(1-\sqrt\rho)n$ whenever
$(q-2Q^2+1)t<(Q^2-1)^2b$.

\begin{cor}\label{cor:GS-fixed-field}
Fix a prime power $Q$, let $q=Q^6$, and fix $0\le\varepsilon<Q/(Q+1)$. Consider the codes $\mathcal C_j$ obtained from Theorem~\ref{thm:GS-fixed-field} at each level $j$ of the Garcia-Stichtenoth tower over $\mathbb F_q$. For all sufficiently large $j$, these codes satisfy $2\g_j-1\le m_j<n_j$, have positive rate and $\g_j=\Theta(n_j)$. Moreover, their counterexample radii satisfy $0<\frac{e_j}{n_j}-\frac{d(\mathcal C_j)}{2n_j}\longrightarrow 0$ and $e_j<(1-\sqrt{\rho_j})n_j.$ Nevertheless, $(1,e_j,\sigma,\varepsilon)$-CA requires $\sigma=1$, while $(1,e_j,\sigma,\varepsilon)$-PG requires $\sigma\ge 1-1/q$.
\end{cor}

\begin{proof}
By Lemma~\ref{lem:GS-tower}, we have $n_j/b_j\to q-1$ and $\g_j/b_j\to Q^3+1$. Since $q-2Q^2+1>2(Q^3+1)$, it follows that $2\g_j-1\le m_j<n_j$ for all sufficiently large $j$. The Riemann--Roch theorem then gives $\dim\mathcal C_j=m_j+1-\g_j$, and hence
\[
\frac{\dim\mathcal C_j}{n_j}
\longrightarrow
1-\frac{2}{Q^4+Q^2+1}-\frac{1}{Q^3-1}>0.
\]
We also have $\g_j=\Theta(n_j)$.

Set $t_j=n_j-(q-1)b_j$. By Lemma~\ref{lem:GS-tower}, we have $t_j\ge1$ and $t_j/b_j\to0$. Theorem~\ref{thm:GS-fixed-field} gives $0<\frac{e_j}{n_j}-\frac{d(\mathcal C_j)}{2n_j}=\frac{t_j}{2n_j}\longrightarrow0.$ Moreover, since $t_j/b_j\to0$, we have $(q-2Q^2+1)t_j<(Q^2-1)^2b_j$ for all sufficiently large $j$. Thus $e_j<(1-\sqrt{\rho_j})n_j$.

Finally, since $e_j=(Q^2-1)b_j+t_j$, we obtain
\[
\frac{(Q^2-Q)b_j}{e_j}
\longrightarrow
\frac{Q^2-Q}{Q^2-1}
=\frac{Q}{Q+1}>\varepsilon.
\]
The soundness lower bounds then follow from Theorem~\ref{thm:GS-fixed-field}.
\end{proof}

For any fixed prime $p$ and $0\le\varepsilon<1$, choose an integer $a\ge1$ such that $\varepsilon<p^a/(p^a+1)$, and set $Q=p^a$ and $q=p^{6a}$. Corollary~\ref{cor:GS-fixed-field} then gives counterexamples over this fixed field with a positive limiting rate, while the relative gap above half the actual minimum distance tends to zero. Thus, for fixed $p$ and $\varepsilon$, the field size can be chosen independently of how small this gap is required to be, while the soundness lower bounds remain unchanged. The limiting rate here is at least $16/21$, whereas Section~\ref{sec:trace-counterexamples} also gives lower-rate examples.

\paragraph{Acknowledgment of AI assistance.}
    ChatGPT 6 was used for language editing and to assist with the results concerning the Johnson bound for one-point AG codes. All other results in this paper were obtained without AI assistance. The authors take responsibility for the correctness of the results and the quality of exposition.

%\appendix
%\input{appendix}

\bibliographystyle{alpha}
\bibliography{refs}

\end{document}